\documentclass[10pt]{article}

\usepackage{subfiles}							
\usepackage[margin=2cm,marginparwidth=2cm,a4paper]{geometry}	
\usepackage{hyperref}							
	\hypersetup{colorlinks,linkcolor={green!30!black},citecolor={blue!50!black},urlcolor={blue!80!black}}
\usepackage{amsmath,amssymb,amsthm}				
	\newtheorem{theorem}{Theorem}
	\newtheorem{lemma}{Lemma}
	\theoremstyle{definition}
	\newtheorem{definition}{Definition}
\usepackage{mathtools}							
\usepackage{IEEEtrantools}						
\usepackage{dsfont}							
\usepackage[backend=bibtex,style=authoryear-icomp,natbib=true,sortcites=true,ibidtracker=strict,sorting=nyvt,giveninits=true,uniquename=init,maxbibnames=99]{biblatex}
\usepackage{tikz}
	\usetikzlibrary{matrix}
\usepackage{graphicx}							
	\graphicspath{{./figures/}}
\usepackage[font=small,labelfont=bf,centerlast]{caption}		
\usepackage{subcaption}							

\newcommand{\ncom}{\newcommand}
\ncom{\set}[2]{\left\{#1\:\middle|\:#2\right\}}				
\ncom{\ket}[1]{\left|#1\right\rangle}						
\ncom{\bra}[1]{\left\langle#1\right|}						
\ncom{\born}[2]{\left|\left\langle#1\middle|#2\right\rangle\right|^2}	
\ncom{\expv}[3]{\left\langle#1\middle|#2\middle|#3\right\rangle}	
\ncom{\cto}{\joinrel\mapsto\joinrel}						
\ncom{\cv}[3]{\begin{pmatrix} #1\\#2\\#3\end{pmatrix}}			
\ncom{\dee}{\,\mathrm{d}}							
\ncom{\pee}{\mathds{P}}								

\DeclareMathOperator{\Forall}{\forall}
\DeclareMathOperator{\supp}{\mathrm{supp}}

\begin{document}


\begin{center}
{\Large\sc{Constraints on Macroscopic Realism\\ Without Assuming Non-Invasive Measurability}}\\
\vspace{1.5em}
{\Large R. Hermens}\\
\emph{Faculty of Philosophy, University of Groningen}\\
\vspace{0.5em}
{\Large O.J.E. Maroney}\\ %
\emph{Faculty of Philosophy and Wolfson College, University of Oxford}\\ %
\vspace{1.5em}
\today
\end{center}

\begin{abstract}
Macroscopic realism is the thesis that macroscopically observable properties must always have definite values.  
The idea was introduced by \citet{LeggettGarg85}, who wished to show a conflict with the predictions of quantum theory, by using it to derive an inequality that quantum theory violates.  
However, Leggett and Garg's analysis required not just the assumption of macroscopic realism \textit{per se}, but also that the observable properties could be measured non-invasively.  
In recent years there has been increasing interest in experimental tests of the violation of the Leggett-Garg inequality, but it has remained a matter of controversy whether this second assumption is a reasonable requirement for a macroscopic realist view of quantum theory.  
In a recent critical assessment \citet{MaroneyTimpson16} identified three different categories of macroscopic realism, and argued that only the simplest category could be ruled out by Leggett-Garg inequality violations.  
\Citet{Allen16} then showed that the second of these approaches was also incompatible with quantum theory in Hilbert spaces of dimension 4 or higher.  
However, we show that the distinction introduced by Maroney and Timpson between the second and third approaches is not noise tolerant, so unfortunately Allen's result, as given, is not directly empirically testable.  
In this paper we replace Maroney and Timpson's three categories with a parameterization of macroscopic realist models, which can be related to experimental observations in a noise tolerant way, and recover the original definitions in the noise-free limit.  
We show how this parameterization can be used to experimentally rule out classes of macroscopic realism in Hilbert spaces of dimension 3 or higher, without any use of the non-invasive measurability assumption.  
Even for relatively low precision experiments, this will rule out the original category of macroscopic realism, that is tested by the Leggett-Garg inequality, while as the precision of the experiments increases, all cases of the second category and many cases of the third category, will become experimentally ruled out.
\end{abstract}

\tableofcontents



\section{Introduction}

The concept of macroscopic realism was introduced by \citet{LeggettGarg85} to focus attention upon an apparent inconsistency between quantum mechanics and our experience of the real world.
Roughly speaking, macroscopic realism maintains that a macroscopically observable property must always have a definite value.
Therefore the only possible states are ones for which macroscopic observables take definite values.
Leggett and Garg further argued that this view could be shown to be inconsistent with observable predictions of quantum theory, by deriving an inequality for the correlations between a sequence of measurements of the macro-observable, that quantum theory could, in principle, violate.

However, Leggett and Garg's derivation required, in addition to macroscopic realism, the use of another assumption: that it was possible, in special cases, to measure the macro-observable non-invasively.
This left open the possibility that a macroscopically realist interpretation of quantum theory is possible that still violates the inequality by denying the possibility of non-invasive measurements.

In recent years there have been increasingly sophisticated experiments seeking to test the Leggett-Garg inequality violations in quantum theory.\footnote{(\citealp{Palacios10}; \citealp{Dressel11}; \citealp{Goggin11}; \citealp{Xu11}; \citealp{Knee12}; \citealp{George13}; \citealp{Knee16}).}
These experiments have used a variety of materials, including superconducting devices, photons, and nuclear and electron spins in silicon and in diamond, and techniques, including weak and strong projective measurements, all confirming the violation.
It should be noted that none of these experiments have actually tested a macroscopically observable property: rather, they have shown violations of the Leggett-Garg inequality for microscopic quantum observables, and represent a proof of principle that tests of the inequality are possible.
Alongside these tests has been a revival of discussion of the significance of the choice of macroscopic realism vs non-invasive measurability.\footnote{(\citealp{FosterElby91}; \citealp{ElbyFoster92}; \citealp{Bacciagaluppi15}; \citealp{Clemente16}; \citealp{Hess16}; \citealp{MaroneyTimpson16}).}

In this paper we follow the analysis of Maroney and Timpson, who identified three types of macroscopic realism, and argued that experiments violating the Leggett-Garg inequality only ruled out one, albeit rather natural, type.
\Citet{Allen16}, building on earlier work by \citet{Allen15}, then showed problems for a second type.
However, the distinction between this type, and the third remaining type as introduced by Maroney and Timpson, is not noise tolerant, and so Allen's result can not be directly subjected to experimental testing.

Our main concern in this paper is to show how experimental tests of macroscopic realism are possible without making use of the non-invasive measurability assumption, and in doing so show that it is possible to rule out a wider class of models than is possible using Leggett-Garg inequality violations.
We start in Section \ref{precon} by using the ontic models formalism (a general framework used for classifying realist interpretations of operational theories) to characterize macroscopically realist models for quantum theory.
By looking at the relationship between macroscopic realism and the eigenvalue-eigenstate link, we will identify the types of macroscopic realism discussed by Maroney and Timpson: here called eigenpreparation mixing models (which are in conflict with Leggett-Garg inequality violations), eigenpreparation supported models, and eigenpreparation undermining models.
Broadly speaking, eigenpreparation mixing is equivalent to macroscopic realism with a strict interpretation of the eigenvalue-eigenstate link, eigenpreparation support keeps a more generalized form of the eigenvalue-eigenstate link, and eigenpreparation undermining models maintain macroscopic realism without a connection to the link.

In Section \ref{QMEU} we review how eigenpreparation mixing and eigenpreparation supported models are incompatible with quantum theory for Hilbert spaces with dimension two or more and three or more respectively, extending the earlier result in \citep{Allen16}.
However, in Section \ref{NoiseT} we show that the distinction between eigenpreparation supported and eigenpreparation undermining macroscopic realism needed for this result, is subject to finite precision loopholes, which means that no experimental test can directly distinguish them.
To address this problem, we introduce two parameters for characterizing macroscopically realist models for quantum theory, which qualitatively distinguish eigenpreparation supported from eigenpreparation undermining models, and which can be tested against experimental data.
With this parameterization, eigenpreparation mixing models can be experimentally ruled out, without using assumptions of non-invasive measurability.
Qualitatively eigenpreparation supported models can also be ruled out, with a larger range of such models being ruled out as experimental precision increases.
For very high precision measurements only some eigenpreparation undermining models remain viable, thus providing a generalization of the noise-free results of Section \ref{OESMRQM} in the limit.
Overall, we will show that a much larger class of macroscopically realist models can be experimentally ruled out than is allowed by Leggett-Garg inequality violations, and without making any use of the assumption of non-invasive measurability.



\section{Macroscopic realism and the eigenvalue-eigenstate link}\label{precon}

Leggett and Garg originally defined macroscopic realism in terms of the existence of macroscopically distinct states:
\begin{quotation}
A macroscopic system with two or more macroscopically distinct states available to it will at all times \textit{be} in one or other of those states. \citep[p. 857]{LeggettGarg85}
\end{quotation}
Intuitively, however, their idea is that it is certain observable \textit{properties}, such as the positions of tables and chairs, that have definite values at all times. This shifts the focus from macroscopic states to macroscopic observables.
That this shift does not alter the meaning should be clear: two states will be macroscopically distinct if, and only if, they assign different values to some macroscopic observable.
But it is non-trivial to make precise why some observables are macroscopic and others are not.\footnote{Although there have been several noteworthy attempts to make the notion of macroscopicity precise. See for example \citep{YadinVedral16} and references therein.}
What we are interested in here is a proof of principle about what kinds of realism about observable properties can be shown to be incompatible with quantum mechanics.
We therefore follow the standard in the literature and set aside the question of what notion of macroscopicity is supposed to be captured by the ``macro''-part.
Although neither the observables we consider here, nor the ones that have been experimentally investigated, fit our intuitive notion of macroscopicity, the results we obtain here do rule out a particular form of realism about these observables.
Whether the results can then be scaled to more \emph{macroscopic} observables is for later concern.
That this is theoretically a possibility is almost trivially so, but whether it is also experimentally possible ultimately relies on what we can technologically achieve and on the ultimate validity of quantum mechanics on the macroscopic scale.

A useful way to approach macroscopic realism is via another idea one often finds in orthodox expositions of quantum mechanics: the eigenvalue-eigenstate link.  This is the axiom that \emph{an observable for a system has a definite value if and only if the system is in an eigenstate for that observable}.
The conjunction of this axiom with the idea that the macro-observables are always value definite yields the requirement that a system is always in one of the eigenstates of a macro-observable.
In other words, the observable imposes a superselection rule: every possible state is a mixture of eigenstates (i.e., a density operator as opposed to a proper superposition).
This is the type of macroscopic realism that \citet[\S3]{MaroneyTimpson16} attribute to \citet{LeggettGarg85}, and can be shown to be ruled out by experimental violations of the Leggett-Garg inequality.

It is important to note that macroscopic realism \emph{per se} is not in conflict with quantum mechanics and so there are principled limitations on what can be shown.
A useful example is the de Broglie-Bohm theory in which all particles have a definite position at all time.\footnote{See also (\citealp{KoflerBrukner13}; \citealp{Bacciagaluppi15}).}
If we assume that macroscopic properties supervene on particle configurations, then these always have well-defined values in this theory.
The problem here is that the relationship between the quantum state of the system and the observable having a definite value in the de Broglie-Bohm theory is different from that in orthodox quantum theory.

Although invoking the eigenvalue-eigenstate link to analyze macroscopic realism is a natural idea within the formalism of quantum mechanics, if one wishes to go beyond this formalism, it doesn't seem that natural anymore.
Because of this it is non-trivial to tease out to what extent results are specifically about quantum mechanics, or whether they also have implications beyond quantum mechanics.
Here, we aim for results that are more theory-independent, in the same spirit that violations of Bell-type inequalities have implications that carry beyond quantum mechanics.
For this we will make use of the ontic models framework (\citealp{Spekkens05}; \citealp{Harrigan07}).
We rely on quantum mechanics solely for inspiration for experimental tests of the ideas introduced, similar to its role in devising experiments in which Bell-type inequalities can be violated.


\subsection{Ontic models}\label{OMsec}


To describe experiments in an operational, theory-independent way, we make use of \emph{Prepare-Transform-Measure (PTM) models} (see also \citep[\S8]{Leifer14}).
A PTM model is a triple $(\mathcal{P},\mathcal{T},\mathcal{M})$ of three sets.
Elements of $\mathcal{P}$ represent possible preparations of the system and provide an operational state description.
Elements of $\mathcal{T}$ represent transformations, i.e., every $T\in\mathcal{T}$ is a function $T:\mathcal{P}\to\mathcal{P}$.
And finally, the elements $M\in\mathcal{M}$ represent measurements.
Specifically, with every measurement $M$ is associated a measurable space $(\Omega_M,\Sigma_M)$, with $\Omega_M$ the set of possible outcomes for the measurement $M$, and a rule which assigns to every $P\in\mathcal{P}$ a probability measure $\pee(\:.\:|M,P)$ over $(\Omega_M,\Sigma_M)$.
We then write
\begin{equation}
	\pee\left(E\middle|M,T,P\right)=\pee\left(E\middle|M,T(P)\right)
\end{equation}
to denote the probability of finding an outcome in $E$ upon a measurement $M$ after the system has been prepared according to $P$ and then transformed according to $T$.

Quantum mechanics can be used to provide PTM models in a straightforward way.
We can take $\mathcal{P}$ to be a set of quantum states, $\mathcal{T}$ a set of unitary operators, and $\mathcal{M}$ a set of self-adjoint operators.
The probabilities are then simply given by the Born rule:
\begin{equation}
	\pee\left(E\middle|A,U,\ket{\psi}\right)=\expv{\psi}{U^*P_A^E U}{\psi},
\end{equation}
where $P_A^E$ is the projection on the subspace spanned by the eigenstates of $A$ for eigenvalues in $E$.
If we assume the projection postulate, then the measurement itself also induces a transformation of the system $\ket{\psi}\mapsto P_A^E\ket{\psi}$.
These transformations can of course be added to $\mathcal{T}$.

To study a particular type of explanation for some feature of a PTM model, we look at ontic models for the PTM model.
An ontic model consists of a measurable space $(\Lambda,\Sigma)$ (where $\Lambda$ is the set of ontic states) and a triplet $(\Pi,\Gamma,\Xi)$ which serves as the counterpart for the triplet $(\mathcal{P},\mathcal{T},\mathcal{M})$ in the following way:
\begin{itemize}
\item $\Pi$ is a set of probability measures on $(\Lambda,\Sigma)$ such that for every $P\in\mathcal{P}$ there is a non-empty subset $\Pi_P\subset\Pi$ of probability measures corresponding to $P$: whenever the system is prepared according to $P$, an ontic state is selected according to some probability measure $\mu\in\Pi_P$.
\item $\Gamma$ is a set of Markov kernels\footnote{For two measurable spaces $(\Omega_1,\Sigma_1),(\Omega_2,\Sigma_2)$, a Markov kernel from the first to the second is a map $\gamma:\Sigma_2\times\Omega_1\to[0,1]$ such that for every $\omega_1\in\Omega_1$ the map $\Delta_2\mapsto\gamma(\Delta_2|\omega_1)$ is a probability measure over $(\Omega_2,\Sigma_2)$ and for every $\Delta_2\in\Sigma_2$ the map $\omega_1\mapsto\gamma(\Delta_2|\omega_1)$ is a measurable function on $(\Omega_1,\Sigma_1)$.} from $(\Lambda,\Sigma)$ to itself such that for every transformation $T\in\mathcal{T}$ there is a non-empty subset $\Gamma_T\subset\Gamma$ of Markov kernels corresponding to $T$: for every $\gamma\in\Gamma_T$ and $\mu\in\Pi_P$ we have $\mu_\gamma\in\Pi_{T(P)}$, where $\mu_\gamma$ is defined as
\begin{equation}
	\mu_\gamma(\Delta):=\int_\Lambda\gamma(\Delta|\lambda)\dee\mu(\lambda)~\Forall\Delta\in\Sigma.
\end{equation}
\item $\Xi$ is a set of Markov kernels such that for every measurement $M\in\mathcal{M}$ there is a non-empty subset $\Xi_M\subset\Xi$.
Every $\xi\in\Xi_M$ is a Markov kernel from $(\Lambda,\Sigma)$ to $(\Omega_M,\Sigma_M)$ such that for every $P\in\mathcal{P}$ and $\mu\in\Pi_P$
\begin{equation}
	\int_\Lambda\xi(E|\lambda)\dee\mu(\lambda)=\pee(E|M,P)~\Forall E\in\Sigma_M,
\end{equation}
i.e., the ontic model reproduces the predictions of the PTM model.
\end{itemize}

In this framework quantum states are associated with preparations of systems: they give rise to probability distributions over ontic states instead of necessarily being ontic states themselves.
It is not excluded that quantum states themselves can be ontic states (in which case their associated probability distribution would assign probability one to itself).
It is just that first and foremost they correspond to probability distributions over ontic states while we remain agnostic about what these ontic states themselves are.

For every PTM model one can construct an ontic model in a trivial way.
This is done by setting $\Lambda=\mathcal{P}$, $\Pi=\set{\delta_P}{P\in\mathcal{P}}$ with $\delta_P$ the Dirac measure peaked at $P$, $\Gamma=\set{\gamma_T}{T\in\mathcal{T}}$ with $\gamma_T(\Delta|P):=\mu_{T(P)}(\Delta)$ and $\Xi=\set{\xi_M}{M\in\mathcal{M}}$ with $\xi_M(E|P):=\pee(E|M,P)$.
This indicates that making use of the ontic models framework is quite a sparse assumption.
Indeed, its name is a bit misleading since there is nothing in the formalism \emph{per se} that requires one to adopt an ontological interpretation of the ontic states.
There is nothing to prevent one from interpreting the elements of $\Lambda$ as, say, a mathematical representative of the degrees of belief of a rational agent.
However, further constraints that may be imposed on the model are usually motivated with an ontological interpretation in mind, and may be artificial if one adopts a different interpretation.

We will note one feature of our use of PTM models in this paper.
Despite their usual description as an operational framework, suggesting they can be related to experimental procedures, following \citet{Spekkens05} it is common to find discussions of the `operational equivalence' of certain preparations, transformations or measurements, particularly on the topic of quantum contextuality.
Such discussions rely on quantifications over all $P\in\mathcal{P}$, $T\in\mathcal{T}$ or $M\in\mathcal{M}$ where $\mathcal{P}$, $\mathcal{T}$ or $\mathcal{M}$ applies to the whole of quantum theory.
Obviously no experimental procedure could ever test all possible preparations allowed by quantum theory (there is a continuous infinity of such procedures), so such formulations cannot be related to actual experimental tests.
Here, our intention is to consider experimental tests of macroscopic realism, for which only a finite fragment of quantum mechanics (i.e., a finite set of preparations, transformations and measurements) can be contemplated.
We will discuss some implications of this restriction in Section \ref{Discus}, including the extent to which it may present loopholes.


\subsection{Macroscopic realism \emph{per se}}

The constraint on ontic models we are concerned with here is one that aims to capture the idea of macroscopic realism.  We follow \citet{MaroneyTimpson16} in identifying this as the constraint that the macro-observable has a definite value at all times, encoded in the ontic states of the system.
Thus, for an observable $Q\in\mathcal{M}$ to be a macro-observable in the ontic model, we require that
\begin{equation}
	\xi(E|\lambda)\in\{0,1\}~\Forall\xi\in\Xi_Q,E\in\Sigma_Q.
\end{equation}
Here we have introduced the notational convention to use $Q$ for measurements of macro-observables.
Now, although value definiteness for all ontic states is a necessary requirement for an observable to be a macro-observable, it is not sufficient.
Value definiteness alone still allows a peculiar form of contextuality.
Depending on the way $Q$ is measured, it may have distinct definite values when there are Markov kernels $\xi,\xi'\in\Xi_Q$ with $\xi\neq\xi'$.
Therefore, we require that in addition $\Xi_Q$ contains precisely one element denoted $\xi_Q$.\footnote{Our definition of macro-observables is similar to the one proposed by \citet{MaroneyTimpson16}, but is  weaker. Maroney and Timpson quantified non-contextuality over all possible measurements, whereas here, in line with our restrictions, we will consider non-contextuality only with respect to measurements within the finite PTM model.  Both can be characterized as ``non-contextual value definiteness for a preferred observable''.}
So a measurement $Q\in\mathcal{M}$ is said to be a \emph{macro-observable in an ontic model} if $\Xi_Q=\{\xi_Q\}$ and $\xi_Q$ only takes the values 0 and 1.
With macroscopic realism \emph{per se} we mean nothing more than that there exists a non-trivial macro-observable.\footnote{A trivial observable would be a measurement that always yields the same outcome with probability one irrespective of how the system is prepared.}

\subsection{Eigenpreparation mixing models}\label{OEMsec}

For the eigenvalue-eigenstate link to appear in this formalism 
%
we need to generalize the idea of eigenstates to eigenpreparations.
A preparation $P\in\mathcal{P}$ is called an \emph{eigenpreparation} for a measurement $M\in\mathcal{M}$ if there is an $m\in\Omega_M$ such that $\pee(m|M,P)=1$.
Imposing macroscopic realism by the introduction of a superselection rule then amounts to the assumption that all possible preparations are mixtures of eigenpreparations of the macro-observable $Q$.
Thus preparations correspond to convex combinations of the probability distributions corresponding to eigenpreparations.
For this reason we say that type of macroscopic realism is \emph{eigenpreparation mixing}.\footnote{\Citet{MaroneyTimpson16} used a more cumbersome terminology. Instead of eigenpreparations they speak of operational eigenstates, and eigenpreparation mixing is referred to as operational eigenstate mixture macrorealism.}


Now consider a PTM model and a macro-observable $Q\in\mathcal{M}$ with possible values $\Omega_Q=\{q_1,\ldots,q_n\}$.
An ontic model for the PTM model is eigenpreparation mixing if every probability measure $\mu\in\Pi$ can be written as a mixture of probability measures corresponding to eigenpreparations.
If there exists such an ontic model, we will also call the PTM model eigenpreparation mixing.

In terms of the PTM model, we find that eigenpreparation mixing places severe constraints.
Because the probability distributions in the ontic model are required to reproduce the predictions of the PTM model, the mixing constraint immediately poses relations on the probability distributions for possible measurements.
For example, if we assume that for every $q_i$ there is precisely one eigenpreparation $P_{q_i}$, we find that every preparation $P$ can be written as a convex combination
\begin{equation}
	P=\lambda_1P_{q_1}+\ldots+\lambda_nP_{q_n}.
\end{equation}
So for every measurement $M$ we have
\begin{equation}\label{convexptm}
	\pee(m|M,P)=\lambda_1\pee(m|M,P_{q_1})+\ldots+\lambda_n\pee(m|M,P_{q_n})
\end{equation}
for all $m\in\Omega_M$.

The prime example of an eigenpreparation mixing model would be orthodox quantum theory where macroscopic realism is enforced by introducing a superselection rule to adhere to the eigen\-value-eigen\-state link.
More generally, any PTM model that satisfies the criteria of Leggett and Garg (macroscopic realism \emph{per se} and non-invasive measurability), is eigenpreparation mixing.

To see this, consider a measurement of a macro-observable $Q$, with values $\Omega_Q=\{q_1,\ldots,q_n\}$, on a system prepared according to an arbitrary preparation $P$.  The post-measurement preparation $P_{q_i}$ is defined by selecting only those cases when the measurement produces the outcome $q_i$.

Within the ontic models framework, macroscopic realism implies a partition of the set of ontic states given by
\begin{equation}\label{partition}
	\Lambda_{q_i}:=\set{\lambda\in\Lambda}{\xi_Q(q_i|\lambda)=1},~i=1,\ldots,n.
\end{equation}
If the outcome of the measurement is $q_i$, then the ontic state before the measurement must have been in the appropriate partition, $\lambda \in \Lambda_{q_i}$. If the measurement is non-invasive, the ontic state after the measurement must still be in that partition.  It follows that the post-measurement preparation $P_{q_i}$ is an eigenpreparation of $Q$ with value $q_i$.  As each post-measurement outcome occurs with probability $\pee(q_i|Q,P)$, the non-selective post-measurement preparation is just: 
\begin{equation}
	P_Q=\sum_{i=1}^n\pee(q_i|Q,P)P_{q_i},
\end{equation}
But, again, the measurement was non-invasive, so $P=P_Q$, which implies that the arbitrary preparation $P$ is eigenpreparation mixing.  Consequently, any test that rules out eigenpreparation mixing models, rules out Leggett Garg type macroscopic realism as well.

Note there is a related proof, that if an observable can be measured non-invasively and \emph{repeatably}, then eigenpreparation mixing must hold.  If the measurement is repeatable, so that if it is performed twice in rapid succession the same outcome always occurs both times, then it is straightforward that the post-measurement preparation $P_{q_i}$ must be an eigenpreparation of $Q$ with value $q_i$.  Combined with non-invasive measurability, eigenpreparation mixing follows as above.  A form of this related proof appears in \citet{ClementeKofler15}, as a proof of macroscopic realism \emph{per se} from non-invasive measurability alone.  However, they do not appear to notice that the repeatability assumption, while a standard assumption within quantum theory, does not follow from non-invasive measurability unless macroscopic realism has been presupposed.


\subsection{Eigenpreparation supported and undermining models}\label{OESSESsec}

Eigenpreparation mixing is a stronger assumption than is needed, if the aim is to preserve the flavor of the eigenvalue-eigenstate link.
In the ontic models formalism it is the ontic state that gives the macroscopic observable a definite value, not the preparation.
In fact, it is not clear what it means for a system to be in an eigenpreparation.
Instead, we should refer to the set of ontic states that may obtain given a particular preparation.

If the model is eigenpreparation mixing, it can be seen to satisfy the following weaker constraint:
\begin{quote}
\textit{The observable $Q$ has the value $q$ if and only if the ontic state of the system could have been obtained by an eigenpreparation with value $q$.}
\end{quote}
We call this the \emph{generalized eigenvalue-eigenstate link}.
Combined with macroscopic realism it implies that the eigenpreparations determine the full set of possible ontic states.
However, we now find that a superselection rule is a more draconian measure than needed to maintain the generalized eigenvalue-eigenstate link, and the resulting restriction to eigenpreparation mixing models is too narrow.
The link does not require that preparations can be written as mixtures of eigenpreparations, but merely that the support of the corresponding probability distributions is a subset of the union of the supports of all the eigenpreparations.
This is still a non-trivial constraint, and we call an ontic model that satisfies this type of macroscopic realism \emph{eigenpreparation supported}.
If the model does not satisfy it, we say it is \emph{eigenpreparation undermining}.\footnote{In \citep{MaroneyTimpson16}, these versions of macroscopic realism are called ``operational eigenstate support macrorealism'' and ``supra eigenstate support macrorealism'' respectively.}
Analogously, we say that a PTM model is eigenpreparation supported if it admits an eigenpreparation supported ontic model and we say that the PTM model is eigenpreparation undermining if it does not admit an eigenpreparation supported ontic model.


It is worthwhile to elaborate a bit more on these definitions.  As noted above, macroscopic realism about $Q$ partitions the ontic state space into the sets $\Lambda_q$.  Given an eigenpreparation $P_{q}$, it is a requirement that every probability measure $\mu\in\Pi_{P_{q}}$ assigns probability 1 to the set $\Lambda_{q}$.
Thus we trivially have one half of the generalized eigenvalue-eigenstate link: if the system is prepared according to an eigenpreparation, then $Q$ has a corresponding definite value.
The converse of this conditional is obviously not true: $Q$ always has a definite value, but not all preparations need be eigenpreparations.
But one can still wonder if, given a particular ontic state $\lambda\in\Lambda_q$ the system could have been prepared according to some eigenpreparation $P_q$.
The generalized eigenvalue-eigenstate link states that this must indeed be so, and rules out the existence of ontic states that can only arise in other preparations.
In other words, if a set of ontic states has probability zero for all eigenpreparations, then it should have probability zero for all preparations, i.e.,
\begin{equation}
	\Forall\Delta\in\Sigma:~\mu(\Delta)=0\Forall\mu\in\Pi_q,\Forall q\in\Omega_Q \implies \mu(\Delta)=0\Forall\mu\in\Pi_P,\Forall P\in\mathcal{P}.
\end{equation}
If we assume a background measure over $(\Lambda,\Sigma)$ with respect to which all probability distributions have a density, then there is a convenient way to reformulate this idea.
For a probability measure $\mu$ let $f_\mu$ denote its density.\footnote{We gloss here over the detail that $f_\mu$ is only determined up to a measure zero set. The reformulation in terms of densities mainly serves to paint a picture, and we expect readers with qualms about technical details to fill them in themselves.}
The support is defined as
\begin{equation}
	\supp f_\mu:=\set{\lambda\in\Lambda}{f_\mu(\lambda)>0}.
\end{equation}
Eigenpreparation support now translates to the claim that every ontic state should be in the support of some density function for some eigenpreparation, i.e.,
\begin{equation}
	\Forall P\in\mathcal{P},\Forall \nu\in\Pi_P:~\supp f_\nu\subset\bigcup_{q\in\Omega_Q}\bigcup_{\mu\in\Pi_q}\supp f_\mu.
\end{equation}

\begin{figure}[ht] 
\begin{center}
\begin{subfigure}{0.23\textwidth}
	\includegraphics[width=\textwidth]{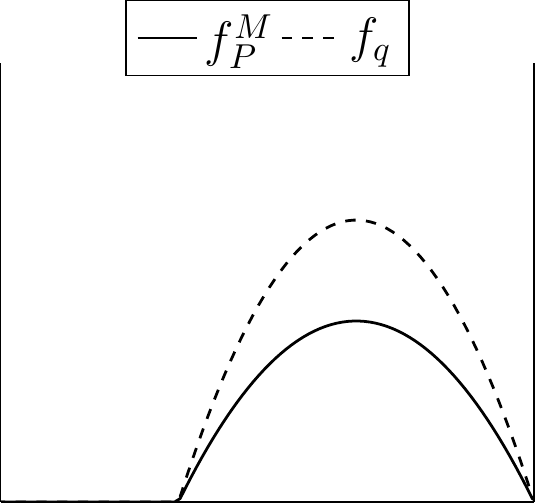}
	\caption{\label{EMMR}}
\end{subfigure}
\begin{subfigure}{0.23\textwidth}
	\includegraphics[width=\textwidth]{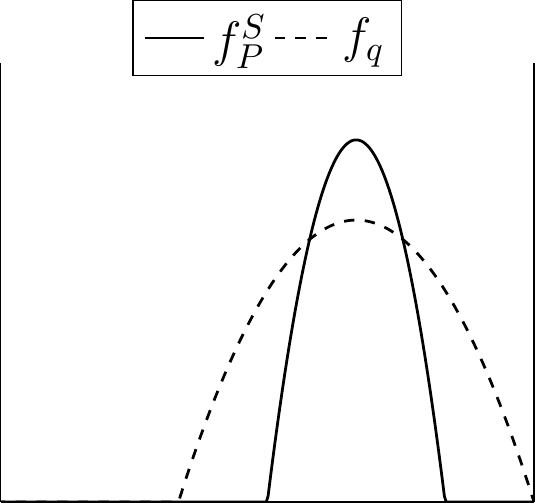}
	\caption{\label{ESMR}}
\end{subfigure}
\begin{subfigure}{0.23\textwidth}
	\includegraphics[width=\textwidth]{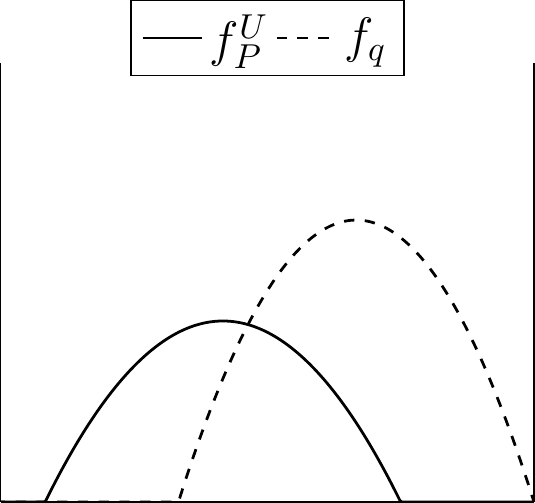}
	\caption{\label{EUMR}}
\end{subfigure}
\begin{subfigure}{0.23\textwidth}
	\includegraphics[width=\textwidth]{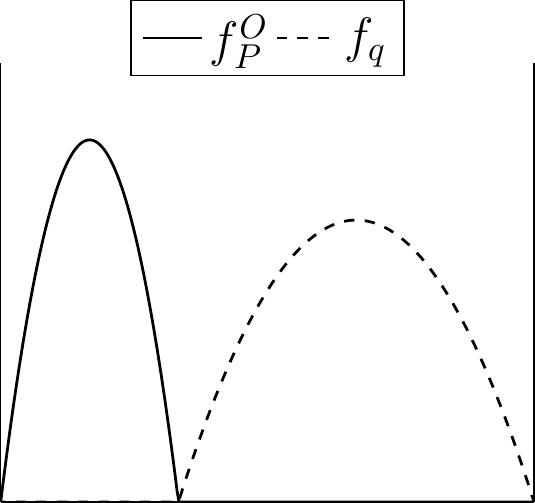}
	\caption{\label{PO}}
\end{subfigure}
\caption{Schematic representation of some ontic models. In each figure the density for an eigenpreparation ($f_q$) and for an arbitrary preparation ($f_P$) are drawn, restricted to the set of ontic states where $Q=q$. The models in \ref{EMMR} and \ref{ESMR} are eigenpreparation supported, while the other two are eigenpreparation undermining.\label{MRplots}}
\end{center}
\end{figure}

In Figure \ref{MRplots} some useful illustrations are given to clarify the definitions.
Here we zoom in on the set of ontic states $\Lambda_q$, which assign the value $q$ to $Q$.
Note that, by definition, the dichotomy eigenpreparation supported/undermining is exhaustive.
The distributions in Figures \ref{EMMR} and \ref{ESMR} are eigenpreparation supported while the distributions in Figures \ref{EUMR} and \ref{PO} are eigenpreparation undermining.
The eigenpreparation mixing models are just special cases of the eigenpreparation supported models.
But they have to adhere to the additional constraint that $f_P^M$ is just a scaled version of $f_q$ on $\Lambda_q$.
On the other extreme end there are models for which the supports are completely disjoint as in Figure \ref{PO}.
This is the case for $\psi$-ontic models, since having disjoint supports is precisely what it means for a model to be $\psi$-ontic \citep{Leifer14}.

An example of an eigenpreparation supported model is the qubit model of Kochen and Specker (\cite*{KS67}, see also \citep[\S4.3]{Leifer14}).
The de Broglie-Bohm theory, being $\psi$-ontic, is an example of an eigenpreparation undermining macro-realist theory.


\section{Quantum theory and eigenpreparation support.}\label{QMEU}

It is now well understood that macroscopic realism with non-invasive measurability is in conflict with quantum theory.
Given the close connection to eigenpreparation mixing models, it should come as no surprise that this type of macroscopic realism is also in conflict with quantum theory.
However, the details are not entirely trivial and therefore we clarify them in section \ref{QMnotEM}.

It is less known that eigenpreparation supported models are also in conflict with quantum mechanics.
This was shown by \citet{Allen16} for quantum systems with Hilbert spaces of dimension greater than 3.
In section \ref{OESMRQM} we show that this also applies to 3-dimensional Hilbert spaces.

\subsection{Quantum mechanics cannot be eigenpreparation mixing}\label{QMnotEM}
Eigenpreparation mixing models can easily be seen to be at odds with quantum mechanical predictions, using the double slit experiment (see also \S10 in the prepreint version of \citep{MaroneyTimpson16}).
If we assume that there is always a fact of the matter considering which slit the particle goes through (i.e. assume it is a macro-observable) we find a tension with eigenpreparation mixing models.
We can consider the experiment with either of the slits closed to be an eigenpreparation for states in which the particle goes through the open slit.
The measurement $M$ in \eqref{convexptm} can be taken to be the measurement of the position where the particle hits the screen.
The interference pattern we observe with both slits open is famously not a convex sum of the two patterns with one slit closed.
Hence, the preparation with two slits open is not a convex combination of the two eigenpreparations.

It is worthwhile to delve a bit more into eigenpreparation mixing models to elucidate their relation to the Leggett-Garg inequality.
Consider again a macro-observable $Q$ with possible values $\Omega_Q=\{q_1,\ldots,q_n\}$.
For an eigenpreparation $P_{q_i}$, the set of corresponding probability distributions is given by $\Pi_{P_{q_i}}$.
Now let $\Pi_{q_i}$ denote the union of all $\Pi_{P_{q_i}}$ with $P_{q_i}$ an eigenpreparation for the value $q_i$.
An eigenpreparation mixing model now requires that for every preparation $P\in\mathcal{P}$, the set $\Pi_P$ lies within the convex hull of $\bigcup_{i=1}^n\Pi_{q_i}$.
As noted above, this leads to non-trivial constraints for the PTM model.
But the situation becomes even more troublesome when transformations are in play.
As an example consider a PTM model with macro-observable $Q$ with $\Omega_Q=\{-1,1\}$ (as in the Leggett-Garg setting).
Suppose for each of the values for $Q$ there is only one eigenpreparation.
Eigenpreparation mixing now implies that for every preparation $P$ and every probability measure $\mu_P\in\Pi_P$ there are measures $\mu_-\in\Pi_{-1},\mu_+\in\Pi_{+1}$ such that
\begin{equation}
	\mu_P=\lambda_P^-\mu_-+\lambda_P^+\mu_+,
\end{equation}
for some positive reals $\lambda_P^-,\lambda_P^+$ that satisfy $\lambda_P^-+\lambda_P^+=1$.

Now consider any transformation $T$ and let $\gamma_T\in\Gamma_T$.
We then have for any measurement $M$ and $E\in\Sigma_M$ that
\begin{equation}
\begin{split}
	\pee(E|M,T,P)
	&=
	\int_\Lambda\int_\Lambda\xi_M(E|\lambda)\gamma_T(\dee\lambda|\lambda')\dee\mu_P(\lambda)\\
	&=
	\int_\Lambda\int_\Lambda\xi_M(E|\lambda)\gamma_T(\dee\lambda|\lambda')\lambda_P^-\dee\mu_-(\lambda)
	+
	\int_\Lambda\int_\Lambda\xi_M(E|\lambda)\gamma_T(\dee\lambda|\lambda')\lambda_P^+\dee\mu_+(\lambda)\\
	&=
	\lambda_P^-\pee(E|M,T,P_-)+\lambda_P^+\pee(E|M,T,P_+).
\end{split}
\end{equation}
It then follows that if two transformations $T_1,T_2$ have the same action on all the eigenpreparations, they must have the same action on all preparations.
That is,
\begin{equation}
	T_1(P_-)=T_2(P_-),~T_1(P_+)=T_2(P_+) \implies T_1=T_2.
\end{equation}
In the special case that $T(P_-)=P_-$ and $T(P_+)=P_+$ we find that $T(P)=P$ for all $P$.
So eigenpreparation mixing implies that, if the effect of a measurement (understood as a transformation) is non-invasive for eigenpreparations, it is non-invasive \emph{tout court}.
However, we know that in quantum mechanics this does not hold.
If a system is prepared in an eigenstate for $Q$, then a measurement of $Q$ does not alter the state.
On the other hand, if the system is prepared in a superposition, then a measurement causes the system to go into one of the eigenstates.
Leggett-Garg tests utilize this feature of quantum mechanics, and so can be used to rule out eigenpreparation mixing models.\footnote{See also (\citealp{MaroneyTimpson16}; \citealp{Knee16}). Another example of this feature is the three box paradox \citep{Maroney16}.}

\subsection{Quantum mechanics is eigenpreparation undermining}\label{OESMRQM}

In Section \ref{precon} three types of macroscopic realism have been introduced: eigenpreparation mixing, eigenpreparation supported, and eigenpreparation undermining.
The first is incompatible with quantum mechanics while the third is compatible with quantum mechanics.
Whether quantum mechanics is eigenpreparation supported has been an open question until quite recently.
\Citet{MaroneyTimpson16} noted that, for 2-dimensional systems, the model introduced by \citet{KS67} serves as an example of an eigenpreparation supported model, while \citet{Allen16} showed that no such model exists if the dimension is greater than 3.

In this section we give a simplified sketch of a proof to show how even in the 3-dimensional case quantum mechanics does not admit an eigenpreparation supported model.  
The theorem presented here can be obtained rigorously as a special case of Theorem \ref{mainthm} in section \ref{THMsec}, for which a detailed proof is given in the appendix.

We proceed in two steps.
First we establish a constraint on eigenpreparation supported PTM models for systems with measurements with at least three distinct outcomes.
Second, we provide an example of a quantum PTM model that violates the constraint derived in Theorem \ref{stelling1}.

\begin{theorem}\label{stelling1}
Let $(\mathcal{P},\mathcal{T},\mathcal{M})$ be a PTM model with macro-observable $Q$ with possible measurement outcomes $\{q_1,q_2,q_3\}$, a second observable $A$ with possible measurement outcomes $\{a_1,a_2,a_3\}$, eigenpreparation $P_{q_1}$ and transformation $T$ such that
\begin{subequations}\label{premis}
\begin{equation}\label{premisa}
	\pee(a_2|A,P_{q_1})=0,
\end{equation}
\begin{equation}\label{premisb}
	\pee(a_3|A,T,P_{q_1})=\pee(q_3|Q,T,P_{q_1})=0.
\end{equation}
\end{subequations}
If an eigenpreparation supported model exists, and $P_{q_1}$ is the only eigenpreparation for $q_1$, then every preparation $P\in\mathcal{P}$ must satisfy
\begin{equation}\label{thmineq}
	\pee(q_1|Q,P)\leq\pee(q_2|Q,T,P)+\pee(a_1|A,T,P).
\end{equation}
\end{theorem}

\begin{proof}[Proof sketch.]
To simplify the proof we will make an additional assumption that the ontic model is also value definite for the observable $A$ (this assumption is not required for the general proof given in the appendix).
The ontic state space can then be partitioned into sets of the form $\Lambda_{q_i}\cap\Lambda_{a_j}$ where $Q$ has the value $q_i$ and $A$ has the value $a_j$.
In this case, every probability measure $\mu$ on the ontic state space gives rise to a probability distribution over the pairs of values $(q_i,a_j)$ for the observables $Q,A$.
This distribution can be neatly summarized in a table:

\begin{equation}
\begin{tikzpicture}[baseline=(current bounding box.center),scale=0.7]
	\draw (0,0) grid (3,3);
	\matrix (mapa) [matrix of nodes, column sep={0.7cm,between origins}, row sep={0.7cm,between origins},
        	every node/.style={minimum size=6mm}, anchor=8.center,	ampersand replacement=\&] at (2.5,2.5)
		{
 		|(1)|  $\mu$ \& |(2)| $\Lambda_{a_1}$ \& |(3)| $\Lambda_{a_2}$ \& |(4)| $\Lambda_{a_3}$ \& \\
 		|(5)|  $\Lambda_{q_1}$ \& |(6)|  $p_1$ \& |(7)| $p_2$ \& |(8)| $p_3$ \& \\
 		|(9)|  $\Lambda_{q_2}$ \& |(10)|  $p_4$ \& |(11)| $p_5$ \& |(12)| $p_6$ \& \\
 		|(13)|  $\Lambda_{q_3}$ \& |(14)|  $p_7$ \& |(15)| $p_8$ \& |(16)| $p_9$ \& \\
		};
\end{tikzpicture}
\end{equation}
Here the $p_i$'s denote the probabilities for the set of states with the specified values for $Q$ and $A$, e.g. $p_4=\mu(\Lambda_{q_2}\cap\Lambda_{a_1})$.

From \eqref{premisa} we conclude that for any measure $\mu\in \Pi_{P_{q_1}}$ we have that all the $p_i$ are zero except for $p_1$ and $p_3$.
Eigenpreparation support then requires that
\begin{equation}\label{concl1}
	\Lambda_{q_1}\cap\Lambda_{a_2}=\varnothing.
\end{equation}
Then $p_2=0$ for arbitrary preparations.
From \eqref{premisb} we derive that for all $\mu\in\Pi_{T(P_{q_1})}$ we have 
\begin{equation}
	p_3=p_6=p_7=p_8=p_9=0.
\end{equation}
Now, let $\gamma\in\Gamma_T$ and for any two subsets of ontic states $\Delta_1,\Delta_2$ let $\Delta_1\stackrel{\gamma}{\cto}\Delta_2$ denote the set of states in $\Delta_1$ that have a finite probability of ending up in $\Delta_2$ under $\gamma$, i.e.,
\begin{equation}
	\Delta_1\stackrel{\gamma}{\cto}\Delta_2:=\set{\lambda\in\Delta_1}{\gamma(\Delta_2|\lambda)>0}.
\end{equation}
Now eigenpreparation support requires that, for any $\gamma\in\Gamma_T$, all the ontic states in $\Lambda_{q_1}$ must evolve to states that are compatible with the predictions of $T(P_{q_1})$.
This implies that
\begin{equation}
\begin{gathered}
	 \left(\Lambda_{q_1}\cap\Lambda_{a_1}\right)\stackrel{\gamma}{\cto}\Lambda_{a_3}=\varnothing,~
	 \left(\Lambda_{q_1}\cap\Lambda_{a_3}\right)\stackrel{\gamma}{\cto}\Lambda_{a_3}=\varnothing,\\
	 \left(\Lambda_{q_1}\cap\Lambda_{a_1}\right)\stackrel{\gamma}{\cto}\Lambda_{q_3}=\varnothing,~
	 \left(\Lambda_{q_1}\cap\Lambda_{a_3}\right)\stackrel{\gamma}{\cto}\Lambda_{q_3}=\varnothing.
\end{gathered}
\end{equation}
Making use of \eqref{concl1}, we can now conclude that all the ontic states in $\Lambda_{q_1}$ evolve to states in $\Lambda_{q_2}\cup\Lambda_{a_1}$ under the transformation $T$, that is,
\begin{equation}
	\gamma\left(\Lambda_{q_2}\cup\Lambda_{a_1}\middle|\lambda\right)=1~\Forall\lambda\in\Lambda_{q_1}.
\end{equation}
From this \eqref{thmineq} follows.
\end{proof}

So we see that for eigenpreparation supported models there are non-trivial constraints on the evolution of ontic states. 
Quantum mechanics gives predictions that can violate \eqref{thmineq}.
For the two observables we take
\begin{equation}
	Q=\sum_{i=1}^3q_i\ket{q_i}\bra{q_i},~
	A=\sum_{i=1}^3a_i\ket{a_i}\bra{a_i},
\end{equation}
with
\begin{equation}
\begin{IEEEeqnarraybox}[][c]{lll}
	\ket{q_1}=\cv{1}{0}{0},\quad&
	\ket{q_2}=\cv{0}{1}{0},\quad&
	\ket{q_3}=\cv{0}{0}{1},\\
	\ket{a_1}=\frac{1}{6}\sqrt{6}\cv{2}{1}{-1},\quad&
	\ket{a_2}=\frac{1}{2}\sqrt{2}\cv{0}{1}{1},\quad&
	\ket{a_3}=\frac{1}{3}\sqrt{3}\cv{-1}{1}{-1}.
\end{IEEEeqnarraybox}
\end{equation}
We then obtain a violation when the preparation $\ket{\psi}$ and transformation $U$ are given by
\begin{equation}
	U=\begin{pmatrix}
		\tfrac{1}{2}\sqrt{2} & \tfrac{1}{2}\sqrt{2} & 0\\
		\tfrac{1}{2}\sqrt{2} & -\tfrac{1}{2}\sqrt{2} & 0\\
		0 & 0 & 1
	\end{pmatrix},~
	\ket{\psi}=\frac{1}{10}\sqrt{10}\cv{1}{1}{2\sqrt{2}},~
	U\ket{\psi}=\frac{1}{10}\sqrt{5}\cv{2}{0}{4}.
\end{equation}
It is easy to check that \eqref{premis} is satisfied and for the probabilities in \eqref{thmineq} we have
\begin{equation}
	\born{q_2}{U\psi}=\born{a_1}{U\psi}=0,~\text{while}~\born{q_1}{\psi}=\frac{1}{10}>0.
\end{equation}

This demonstrates that eigenpreparation support is in conflict with quantum theory in Hilbert spaces of dimension 3 or more.
As the Kochen-Specker model provides a constructive example of an eigenpreparation supported model for 2 dimensional Hilbert spaces, Theorem \ref{stelling1} closes the logical gap between Kochen-Specker and Allen, Maroney and Gogioso.


\section{Noise tolerance and eigenpreparation support}\label{NoiseT}

In Section \ref{QMEU} we showed that quantum theory must be eigenpreparation undermining, in all Hilbert spaces of dimension greater than two.
In this Section, we will show that, unfortunately, this proof rests on a distinction between eigenpreparation support and eigenpreparation undermining that is not noise tolerant.
The proof uses probabilities which are assumed to be zero \eqref{premis}.
Such an assumption cannot be verified experimentally.
At best one can confirm that these probabilities are small, but the proof really requires them to be zero.
 Worse, we can show that for any finite value of these probabilities, one can devise eigenpreparation supported models that can reproduce the predictions of quantum mechanics.

We will therefore introduce a new way of characterizing eigenpreparation support and eigenpreparation undermining models, $(\alpha,\beta)$-support, that captures the qualitative features of these two types of macroscopic realism in a noise-tolerant way.
In our main result, we will then show how experimentally testable predictions of quantum theory are at odds with all qualitatively eigenpreparation supported models, in Hilbert spaces of dimension greater than 2.

\subsection{Robustness of the supported/undermining distinction}\label{robustsec}

As noted above, Theorem \ref{stelling1} relies on an assumption that cannot be verified experimentally.
This is just a symptom of a deeper problem, namely, that the distinction between eigenpreparation supported and eigenpreparation undermining is not noise-tolerant.
The reason is that the notion of eigenpreparation support implicitly relies on the use of the asymmetric overlap between preparations, which is not robust.
For two probability measures $\mu,\nu$ on $(\Lambda,\Sigma)$ the asymmetric overlap is defined as
\begin{equation}
	\varpi(\nu|\mu):=\inf\set{\mu(\Delta)}{\Delta\in\Sigma,\nu(\Delta)=1}.
\end{equation}
Or, in terms of the corresponding density functions,\footnote{The density functions are defined with respect to a suitably chosen background measure. Here, and throughout the remainder of this paper, this background measure is suppressed in our notation of the integral.}
\begin{equation}
	\varpi(f_\nu|f_\mu):=\int_{\supp f_\nu}f_\mu(\lambda)\dee\lambda.
\end{equation}
If we assume that $\Pi_q$ is convex, then eigenpreparation support results in the criterion that
\begin{equation}
	\sup_{\mu_q\in\Pi_q}\varpi(\mu_q|\mu_P)=\pee(q|Q,P) 
\end{equation}
for all $q\in\Omega_Q$, $\mu_P\in\Pi_P$, and $P\in\mathcal{P}$.

\begin{figure}
\begin{center}
\begin{subfigure}{0.23\textwidth}
	\includegraphics[width=\textwidth]{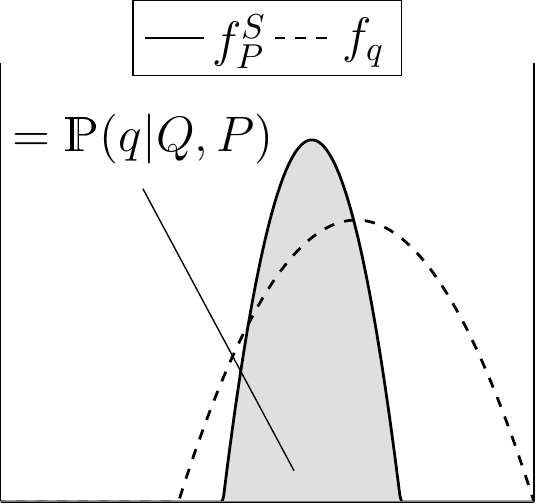}
	\caption{\label{A}}
\end{subfigure}
\begin{subfigure}{0.23\textwidth}
	\includegraphics[width=\textwidth]{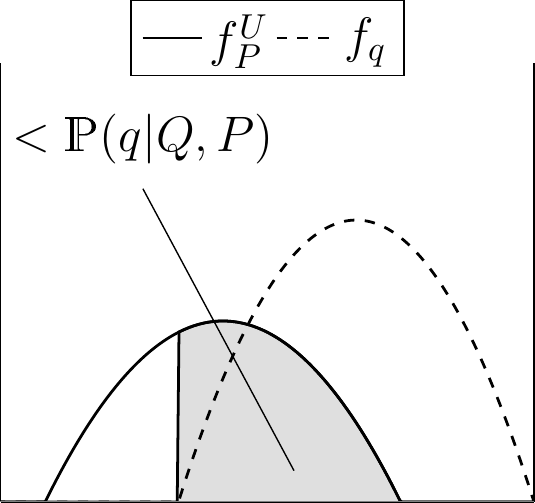}
	\caption{\label{B}}
\end{subfigure}
	
\begin{subfigure}{0.23\textwidth}
	\includegraphics[width=\textwidth]{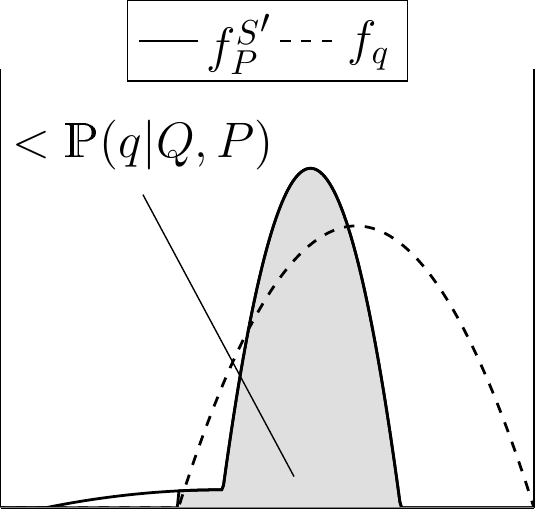}
	\caption{\label{D}}
\end{subfigure}
\begin{subfigure}{0.23\textwidth}
	\includegraphics[width=\textwidth]{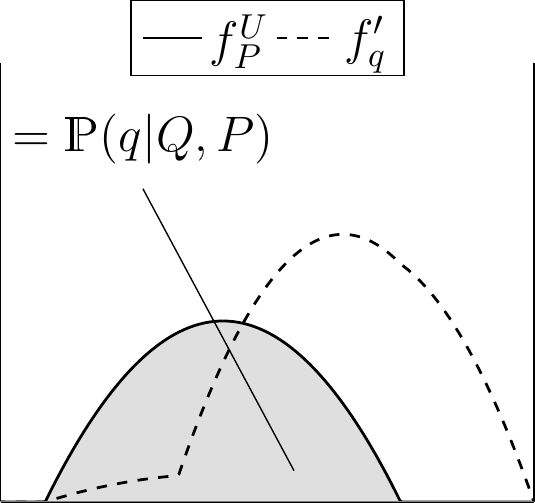}
	\caption{\label{C}}
\end{subfigure}
\caption{Schematic representation of some ontic models. In each figure the density for an eigenpreparation ($f_q$) and for an arbitrary preparation ($f_P$) are drawn, restricted to the set of ontic states where $Q=q$. The colored area corresponds to the asymmetric overlap. The models in \ref{A} and \ref{C} are eigenpreparation supported, while the other two are eigenpreparation undermining.\label{plaatje}}
\end{center}
\end{figure}

Figure \ref{plaatje} provides a good illustration of the problems with the asymmetric overlap.
In this figure we again zoom in on $\Lambda_q$.
The distributions in Figure \ref{A} are in accordance with eigenpreparation support: the asymmetric overlap $\varpi(f_q|f_P^S)$ equals the appropriate probability.
The distributions in Figure \ref{B}, on the other hand, are eigenpreparation undermining: the support of $f_P^U$ extends beyond the support of $f_q$.
However, by slightly modifying the model depicted in Figure \ref{B}, it can be made eigenpreparation supported.
This is done by replacing $f_q$ with
\begin{equation}
	f_q':=(1-\epsilon)f_q+\frac{\epsilon}{\pee(q|Q,P)}f_P^U,
\end{equation}	
as in Figure \ref{C}.
If $\epsilon$ is small enough, the distributions $f_q$ and $f_q'$ are experimentally indiscernible.
Consequently, for any eigenpreparation undermining model and any finite precision of measurements, we can construct a modified model that is eigenpreparation supported and experimentally indiscernible from the original model.
Conversely, every preparation supported model can be modified to obtain a preparation undermining model by slightly adjusting the distribution $f_P^S$ as in Figure \ref{D} where we have replaced it with
\begin{equation}
	{f_P^S}':=(1-\epsilon)f_P^S+\epsilon f_P^U.
\end{equation}

These considerations show that there is no way to experimentally rule out eigenpreparation supported models.
But they also suggest that the distinction between eigenpreparation supported and undermining models wasn't formalized in an operationally meaningful way.  
Just by looking at Figure \ref{plaatje} it is not difficult to convince oneself that the explanations provided for observed phenomena by the ontological model in Figure \ref{D} will be almost identical to those provided by the model in Figure \ref{A}.  
While Figure \ref{D} is, strictly, an eigenpreparation undermining model, it is, qualitatively, almost an eigenpreparation supported model.  
Similarly the model in Figure \ref{C} is, strictly, eigenpreparation supported, but its explanation for observed phenomena will have more in common with Figure \ref{B}.  
It is, qualitatively, almost eigenpreparation undermining.
So we wish to draw a distinction between eigenpreparation supported and undermining models that is based on a matter of gradation between the qualitative features of the models, rather than an all-or-nothing case.
The goal is to make this idea precise in a way that is noise-tolerant, which is what we do in the next section.

\subsection{\texorpdfstring{$(\alpha,\beta)$}{ab}-supported models}\label{LMRsec}

We now introduce our two parameter characterization of ontic models, which will allow us to draw our noise-tolerant distinction between the qualitative features of such models which make them eigenpreparation supported or undermining.  We start by looking at noise-tolerance.

As noted in the previous section, the problem of the robustness of the eigenpreparation supported/un\-der\-min\-ing distinction can be traced back to its reliance on the asymmetric overlap.
An intuitive suggestion then is to switch from the asymmetric overlap to the symmetric overlap, which \emph{is} noise-tolerant.
For two positive measurable functions $f,g$ the symmetric overlap is defined as
\begin{equation}
	\omega(f,g):=\int_\Lambda \min\left(f(\lambda),g(\lambda)\right)\dee\lambda.
\end{equation}
When $f$ and $g$ are densities for two probability measures $\mu_f$ and $\mu_g$, the symmetric overlap can be written as
\begin{equation}\label{symoverl}
	\omega(f,g)
	=1-\sup_{\Delta\in\Sigma}\left|\mu_f(\Delta)-\mu_g(\Delta)\right|
	=1-\frac{1}{2}\int_\Lambda\left|f(\lambda)-g(\lambda)\right|\dee\lambda.
\end{equation}
Operationally this has a well known and robust interpretation: if a system is prepared using one of the two preparation procedures, but one doesn't know which, then the best guess possible as to which procedure was used cannot succeed with a probability higher than $1-\tfrac{1}{2}\omega(f,g)$.

Unfortunately, just switching to the symmetric overlap tells us little about whether a model is more eigenpreparation supported or more eigenpreparation undermining.
Figure \ref{plaatje} reflects this: in \ref{A} and \ref{B} the symmetric overlap is of the same order even though they are supposed to be prime examples of the two types of macroscopic realism.
The problem is that we do not care about the size of the symmetric overlap, per se, but rather about the way the density for the preparation deviates in shape from the eigenpreparation.

To make this idea precise, we return to the case of eigenpreparation mixing (Figure \ref{EMMR}).
In the simple case where for every value of $Q$ there is precisely one eigenpreparation, and only one corresponding probability distribution in the ontic model, we can write
\begin{equation}
	f_P=\pee\left(q_1\middle|Q,P\right)f_{q_1}+\ldots+\pee\left(q_n\middle|Q,P\right)f_{q_n}.
\end{equation}
This is equivalent to the criterion that
\begin{equation}\label{OEMMcrit}
	\omega\left(f_P,\pee\left(q_i\middle|Q,P\right)f_{q_i}\right)=\pee\left(q_i\middle|Q,P\right)
\end{equation}
for all $q_i$.

Essentially, \eqref{OEMMcrit} tells us that eigenpreparation mixing means that $f_P$ is shaped like $f_{q_i}$ on $\Lambda_{q_i}$ for every $i$.
A violation of eigenpreparation mixing can thus be understood as a deviation from the shape of $f_{q_i}$.
Possible deviations can be classified in two categories corresponding to the two other types of macroscopic realism.
We will identify two parameters $\alpha$ and $\beta$ to characterize how much a given $f_P$ deviates in either of these two directions from the eigenpreparation mixing shape.

First we look at eigenpreparation support.  For an eigenpreparation supported model \eqref{OEMMcrit} will not hold in general.
However we can still find numbers $a_i$ such that
\begin{equation}\label{OEMMcrit2}
	\omega\left(f_P,a_if_{q_i}\right)=\pee\left(q_i\middle|Q,P\right),
\end{equation}
but now with $\sum_i a_i >1$.
If we take these $a_i$'s as small as possible, then we can take $\alpha=\max\{a_1,\ldots,a_n\}$ as a measure for how much the model deviates from eigenpreparation mixing in the eigenpreparation support direction.
In the limit where we allow $\alpha$ to go to infinity, we find that all eigenpreparation supported models can be understood as such deviations from eigenpreparation mixing.
This is because we have the following general relation between the symmetric and the asymmetric overlap:
\begin{equation}\label{limitexpr}
	\lim_{\alpha\to\infty}\omega(f_P,\alpha f_q)=\varpi(f_q|f_P).
\end{equation}
The original definition for eigenpreparation support was that the right hand side of \eqref{limitexpr} equals $\pee\left(q_i\middle|Q,P\right)$.
On the new reading we use the left hand side of \eqref{limitexpr} and eigenpreparation support corresponds to a finite value of $\alpha$ for which the equation
\begin{equation}\label{alphadef}
	\omega(f_P,\alpha f_q)=\pee\left(q\middle|Q,P\right)
\end{equation}
holds for all $q$.

Now we look at eigenpreparation undermining models.  These can be understood as a deviation from eigenpreparation mixing where the shape of $f_P$ is altered by extending its domain beyond the support of $f_q$.  
Now this will leave $\omega(f_P,f_q)\leq\varpi(f_q|f_P)<\pee\left(q\middle|Q,P\right)$, due to the support of $f_P$ that now lies outside the support of $f_{q_i}$.  
However, while scaling $f_q$ up to $\alpha f_q$ will eventually achieve $\omega(f_P,\alpha f_q)=\varpi(\alpha f_q|f_P)$, this will not change the value of the asymmetric overlap, as the supports of $f_q$ and $\alpha f_q$ are identical.  
There will remain a part of the support of $f_P$ which lies outside the support of $f_{q_i}$.

There will, however, always exist numbers $b_i$ such that
\begin{equation}
 \omega(f_P,f_{q_i})+b_i = \pee\left(q_i\middle|Q,P\right)
\end{equation}
If we take $\beta=\max\{b_1,\ldots,b_n\}$, the deviation of eigenpreparation undermining models can be quantified in terms of the parameter $\beta$ such that
\begin{equation}\label{betadef}
	\omega(f_P,f_q)+\beta\geq\pee\left(q\middle|Q,P\right).
\end{equation}
holds for all $q$.

Equations \eqref{alphadef} and \eqref{betadef} can be combined to give a more qualitative understanding of eigenpreparation supported and mixing models. 
For any model, there always exist $\alpha$ and $\beta$ such that
\begin{equation}
	\omega(f_P,\alpha f_q)+\beta\geq\pee\left(q\middle|Q,P\right)
\end{equation}
for all $q$. (Indeed, the inequality holds trivially for the choice $\beta=1$.)  Formally we introduce the following definition.

\begin{definition}\label{lmrdef}
Given an ontic model for a PTM model with macro-observable $Q$, preparation $P$ and constants $\alpha\in[0,\infty)$, $\beta\in[0,1]$, we say that $P$ is \emph{$(\alpha,\beta)$-supported on $\Lambda_q$} for some $q\in\Omega_Q$ if
\begin{equation}\label{abdef}
	\sup_{f_q\in\Pi_q}\omega\left(f_P,\alpha f_q\right)+\beta\geq \pee(q|Q,P)~\Forall f_P\in\Pi_P.
\end{equation}
If every preparation is $(\alpha,\beta)$-supported for all $q\in\Omega_Q$ we say that the ontic model is \emph{$(\alpha,\beta)$-supported}.
A PTM model is called \emph{$(\alpha,\beta)$-supported} if it admits an $(\alpha,\beta)$-supported ontic model.
\end{definition}


Now we look at the behavior of eigenpreparation supported and eigenpreparation undermining models in terms of $(\alpha,\beta)$-support.
We focus, as usual, upon a single outcome $q$ with eigenpreparation $f_q$ and we introduce the concept of $(\alpha,\beta)$-support curves.
The $(\alpha,\beta)$-support curve for a preparation $P$ is simply the modified overlap $\omega(f_P,\alpha f_q)$ viewed as a function of $\alpha$.
When $\alpha=0$, the modified overlap $\omega\left(f_P,\alpha f_q\right)=0$ and so $\beta\geq \pee(q|Q,P)$.  
As $\alpha$ is allowed to increase, $\beta$ is able to fall. 
At some point, however, $\omega\left(f_P,\alpha f_q\right)$ reaches a maximum, and $\beta$ reaches its lowest value.  
For eigenpreparation supported models, this value of $\beta=0$. 

These considerations are reflected in Figure \ref{overlapfigure}.
Here we have plotted the $(\alpha,\beta)$-support curves for each of the functions from Figure \ref{plaatje} as well as for the eigenpreparation mixing model (the $\psi$-ontic case is left out because for such models $\omega(f_P,\alpha f_q)=0$ for all values of $\alpha$). 

Next we look at the models in Figures \ref{D} and \ref{C}, as $\alpha$ and $\beta$ vary.  
In the case of Figure \ref{D}, as $\alpha$ is allowed to increase, $\beta$ will fall in much the same manner as for Figure \ref{A}.  
Only as the modified overlap approaches $\pee(q|Q,P)$ will any significant difference arise: Figure \ref{D} will reach a minimum $\beta$ value that is just above zero. 
By contrast, Figure \ref{C} will initially behave much the same as Figure \ref{B}.  
Only when $\alpha$ gets large, and $\beta$ approaches its minimum value for Figure \ref{B}, will a difference appear: for the model in Figure \ref{C}, $\beta$ will continue to fall slowly, as $\alpha$ rises.

\begin{figure}[ht]
\begin{center}
\includegraphics[width=0.6\textwidth]{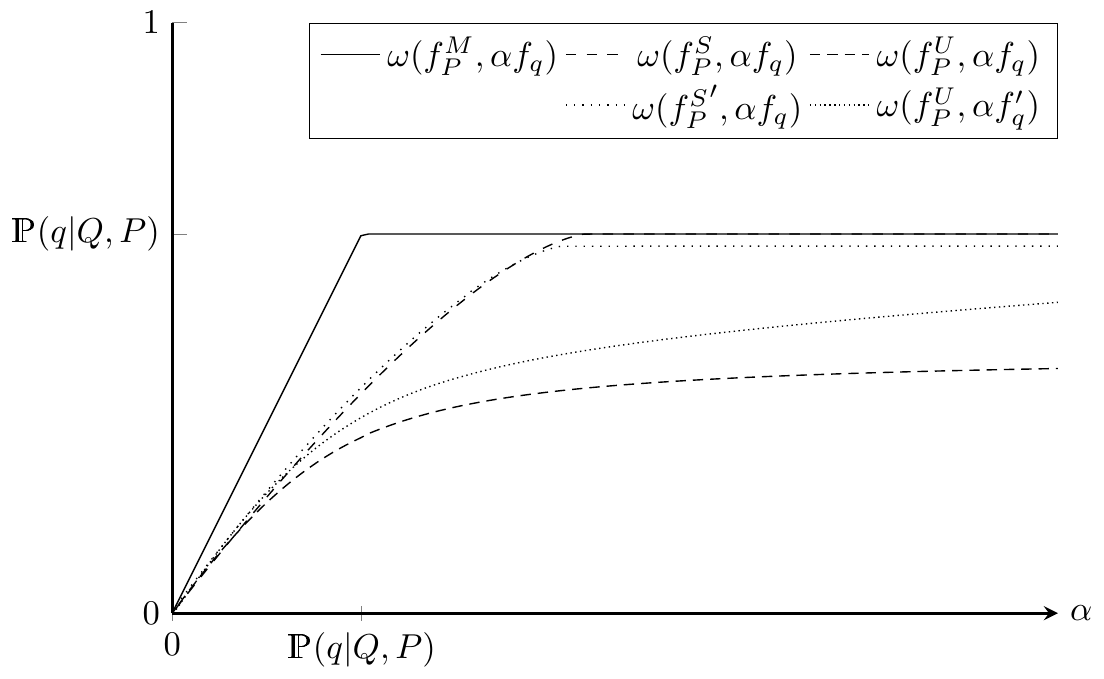}
\end{center}
\caption{Plots for the symmetric overlap as a function of $\alpha$ for the distributions from Figure \ref{plaatje}. Although for all the eigenpreparation supported models the overlap tends towards the maximum value $\pee(q|Q,P)$, this is at a very slow rate for the qualitatively eigenpreparation undermining distribution with $f_q'$. For the qualitatively eigenpreparation supported model with $f_P^{S'}$ on the other hand, the graph tends quite rapidly to its maximum value which is just below $\pee(q|Q,P)$. \label{overlapfigure}}
\end{figure}

We now see a clear qualitative similarity between Figures \ref{A} and \ref{D}, that is not shared with Figures \ref{B} and \ref{C}: $\beta$ can reach relatively small values, for relatively small values of $\alpha$.  While $\beta$ can become small for Figure \ref{C}, it can only do so for quite large values of $\alpha$ (and for Figure \ref{B} not at all).  It is this feature that captures the qualitative properties of the models more meaningfully than the question of whether $\beta$ can, for some sufficiently large value of $\alpha$, actually reach zero.

Broadly, a \emph{qualitatively} eigenpreparation undermining model is one in which there is some region $\Delta$ of the ontic state space, for which $\mu_q(\Delta)$ is very small, but $\mu_P(\Delta)$ is not small.  
In this case there is a significant region of the support of $f_P$ where $f_q$ is much smaller than $f_P$. 
This requires a large $\alpha$ scaling up of $f_q$ to rise above $f_P$, so $\beta$ cannot fall to low values until $\alpha$ has risen to large values.  
In the case of strict eigenpreparation undermining models, $f_q=0$ for this region and no $\alpha$ scaling will ever succeed to make $\beta=0$.

By contrast, in \emph{qualitatively} eigenpreparation supported models, every region $\Delta$  for which $\mu_q(\Delta)$ is small, must also have $\mu_P(\Delta)$ as small. 
There is then nowhere that $f_P$ is significantly larger than $f_q$. 
Such models can still be strictly eigenpreparation undermining when there is a region for which $\mu_P(\Delta)$ is small and $\mu_q(\Delta)= 0$. 
However, as $\mu_P(\Delta)$  is now small, scaling up $\mu_q(\Delta)$ by $\alpha$ does not miss very much of the support of $f_P$, and only a small value of $\beta$ is required to compensate.

What is important to both types of models is not whether $\mu_q(\Delta)= 0$ or is just very small, which is not robust against small changes in the model: it is whether the corresponding value of $\mu_P(\Delta)$ is small or large, which is robust against small changes.  Put in its simplest terms, then, a qualitatively eigenpreparation supported model is one which can be $(\alpha,\beta)$-supported, with both $\alpha$ and $\beta$ remaining relatively low.  A qualitatively eigenpreparation undermining model  is one where $\beta$ cannot become small without $\alpha$ becoming large, if at all.

We may regard this notion as a further weakening of the generalized eigenvalue-eigenstate link. 
This link requires that if the observable has a definite value, then the ontic state lies in the support of an eigenpreparation. 
The qualitatively weakened notion requires only that no region of the ontic state space $\Delta$, assigning a definite value to the observable, can be prepared with $\mu_P(\Delta)$ large, unless $\mu_q(\Delta)$ is also large for some eigenpreparation $q$ of the observable $Q$.  Qualitatively eigenpreparation supported models satisfy this further weakened link, while qualitatively eigenprepation undermining models do not.

Before moving on, a few formal remarks about Definition \ref{lmrdef} are in order.
We have introduced the supremum over all densities in $\Pi_q$ to account for the possibility of multiple eigenpreparations for $q$.
One can imagine that there are two densities $f_q^1$ and $f_q^2$ such that
\begin{equation}
	\omega(f_q^1,f_q^2)=\omega(f_q^1,f_P)=0,~\text{while}~\omega(f_q^2,f_P)=\pee(q|Q,P).
\end{equation}
In the quantum case, for example, this may happen whenever the eigenvalue $q$ is degenerate.
Even though $\omega(f_q^1,f_P)=0$, we still think of this case as eigenpreparation supported, and taking the supremum in \eqref{abdef} takes care of this.\footnote{We implicitly assume here that $\Pi_q$ is convex. To see this, consider the case where $\omega(f_q^1,f_q^2)=0$ and $\omega(f_q^1,f_P)=\omega(f_q^2,f_P)=\tfrac{1}{2}\pee(q|Q,P)$. This also suggests eigenpreparation support, but we see that if $f_q^1$ and $f_q^2$ are the only two distributions in $\Pi_q$, we need $\beta\geq\tfrac{1}{2}\pee(q|Q,P)$. But if we allow convex combinations, we may have $\beta=0$.}

Second, note that, if an ontic model is $(\alpha,\beta)$-supported, then it is also $(\alpha',\beta')$-supported whenever $\alpha'\geq\alpha$ and $\beta'\geq\beta$.
In fact, an ontic model does not give rise to unique values of $\alpha$ and $\beta$ for which it is $(\alpha,\beta)$-supported, but rather gives rise to a region of pairs in the parameter space $[0,\infty)\times[0,1]$.
This region is never empty since every macroscopic realist ontic model is $(\alpha,1)$-supported for all values of $\alpha$.
When restricting attention to a single preparation we can use Figure \ref{overlapfigure} to gain some more insight.
For any $\alpha\in[0,\infty)$ we find that the preparation $P$ is $(\alpha,\beta)$-supported on $\Lambda_q$ if and only if $\beta\geq\pee(q|Q,P)-\omega(f_P,\alpha f_q)$.

Finally, if the model is $(\alpha,0)$-supported for some value of $\alpha$, then it follows that the model is eigenpreparation supported.
Conversely, we see that if the model is eigenpreparation supported and $f_P$ is bounded, then there is always an $\alpha$ such that the model is $(\alpha,0)$-supported.
The requirement that $f_P$ is bounded is a misleading formality though.
After all, whether it is bounded or not depends on the choice of the background measure, but whether the model is $(\alpha,\beta)$-supported or not does not depend on this choice.
Thus, without loss of generality, we may always assume that $f_P$ is bounded (possibly by changing the background measure), and then we recover the statement that the model is $(\alpha,0)$-supported for some value of $\alpha$ if and only if the model is eigenpreparation supported.

\subsection{Constraints on \texorpdfstring{$\alpha$}{a} and \texorpdfstring{$\beta$}{b}}\label{THMsec}

To rule out $(\alpha,\beta)$-support for a given pair $(\alpha,\beta)$, it suffices to focus on a finite fragment of quantum mechanics containing a particular preparation $P$ (associated with a quantum state), and show that any ontic model in which $P$ is $(\alpha,\beta)$-supported makes predictions that contradict the quantum mechanical predictions.
This is actually similar to the case of eigenpreparation support.
Indeed, Theorem \ref{stelling1} only made use of the assumption that the support of $f_P$ is a subset of the supports of the $f_q$'s.
So it also rules out eigenpreparation undermining models as long as they behave as eigenpreparation supported models for the preparation $P$.

As noted in the previous section, if an ontic model is $(\alpha,\beta)$-supported with $\beta=0$, then it is eigenpreparation supported.
In section \ref{OESMRQM} we demonstrated that such models must satisfy an inequality that is significantly violated by quantum mechanics.
It is therefore not surprising that quantum mechanics poses further constraints on the possibility of models that are $(\alpha,\beta)$-supported with positive values of $\beta$.
In fact, Theorem \ref{stelling1} can be generalized to incorporate the notion of $(\alpha,\beta)$-support, leading to the following theorem:

\begin{theorem}\label{mainthm}
Let $(\mathcal{P},\mathcal{T},\mathcal{M})$ be a PTM model with two measurements $Q,A\in\mathcal{M}$ with each three possible measurement outcomes $\{q_1,q_2,q_3\}$, $\{a_1,a_2,a_3\}$ and a single eigenpreparation $P_{q_1}$ for the value $q_1$.
Let $P\in\mathcal{P}$ be any preparation.
For any pair $(\alpha,\beta)\in[0,\infty)\times[0,1]$, if there exists an ontic model for which $P$ is $(\alpha,\beta)$-supported on $\Lambda_{q_1}$, then for any transformation $T\in\mathcal{T}$ the following inequality holds:
\begin{equation}\label{mainineq}
	\pee(q_1|Q,P)-\pee(q_2|Q,T,P)-\pee(a_1|A,T,P)\leq
	\alpha\left(\pee(a_2|A,P_{q_1})+\pee(a_3|A,T,P_{q_1})+\pee(q_3|Q,T,P_{q_1})\right)+\beta.
\end{equation}
\end{theorem}

The proof for this theorem can be found in the appendix.
Theorem \ref{stelling1} is recovered as a special case by assuming that
\begin{equation}\label{unwaras}
	\pee(a_2|A,P_{q_1})=\pee(a_3|A,T,P_{q_1})=\pee(q_3|Q,T,P_{q_1})=0
\end{equation}
and setting $\beta=0$.
The inequality \eqref{mainineq} then reduces to \eqref{thmineq} and the notion of $(\alpha,\beta)$-support reduces to eigenpreparation support.

Even if we were to take \eqref{unwaras} for granted, we find that this theorem improves on Theorem \ref{stelling1}.
This can be seen as follows.
Theorem \ref{stelling1} can be paraphrased as the claim that, if a PTM model predicts that
\begin{equation}\label{PTMpred}
	\pee(q_1|Q,P)>\pee(q_2|Q,T,P)+\pee(a_1|A,T,P),
\end{equation}
then there does not exist an eigenpreparation supported ontic model for the PTM model.
In the new language, this is equivalent to the claim that there does not exist an $(\alpha,0)$-supported ontic model.
From Theorem \ref{mainthm} we can now further conclude that all $(\alpha,\beta)$-supported ontic models are ruled out with
\begin{equation}
	\beta<\pee(q_1|Q,P)-\pee(q_2|Q,T,P)-\pee(a_1|A,T,P)
\end{equation}
if the PTM model predicts \eqref{unwaras}.

\begin{figure}[ht]
\begin{center}
\includegraphics[width=0.6\textwidth]{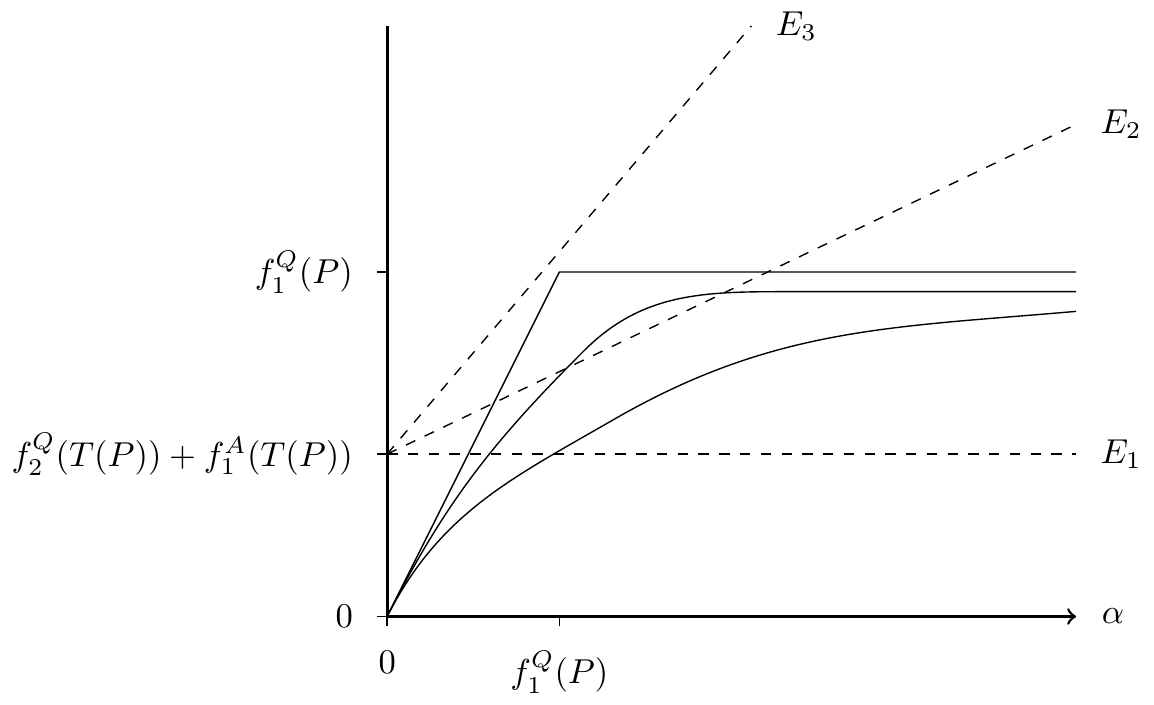}
\end{center}
\caption{Schematic depiction of how experimental results can be used to rule out $(\alpha,\beta)$-supported models.
The solid lines represent ontic models and depict their corresponding $(\alpha,\beta)$-support curves.
The dashed lines indicate constraints obtained by possible experimental data. The slope is given by  $f^A_2(P_{q_1})+f^A_3(T(P_{q_1}))+f^Q_3(T(P_{q_1}))$. An ontic model is ruled out by the experiment experimental data if the symmetric overlap crosses the line $E_i$ derived from that data. The line $E_1$ corresponds to the idealized case where \eqref{unwaras} holds and the line $E_3$ represents an experiment that is too noisy to rule out any models. \label{ab-exp-label}}
\end{figure}

The main problem with Theorem \ref{stelling1} was that \eqref{unwaras} can never be verified experimentally.
Theorem \ref{mainthm} solves this problem by putting constraints on $\alpha$ and $\beta$ even in the presence noise.
Experimentally, one can get estimates for the probabilities in \eqref{mainineq} in the form of relative frequencies for measurement outcomes.
So the experimental estimate for $\pee(m_i|M,P)$ would be $f_i^M(P)\in[0,1]$.
The acquired data set determines a line given by
\begin{equation}
	E(\alpha):= f^Q_2(T(P))+f^A_1(T(P))+\alpha\left(f^A_2(P_{q_1})+f^A_3(T(P_{q_1}))+f^Q_3(T(P_{q_1}))\right).
\end{equation}
In Figure \ref{ab-exp-label} three candidates for such a line are illustrated.
An ontic model is ruled out experimentally whenever $\omega(f_P,\alpha f_{q_1})$ crosses the line $E(\alpha)$ for some value of $\alpha$.
The explanation for this criterion runs as follows.
For any value of $\alpha$, the preparation $P$ is $(\alpha,\beta_\alpha)$-supported for the choice $\beta_\alpha=f^Q_1(P)- \omega(f_P,\alpha f_{q_1})$.
Now suppose there exists an ontic model for which there exists an $\alpha$ such that $\omega(f_P,\alpha f_{q_1})>E(\alpha)$.
We then find for this value of $\alpha$ that
\begin{equation}
\begin{split}
	f^Q_1(P)-f^Q_2(T(P))-f^A_1(T(P))
	={}&
	\beta_\alpha+\omega(f_P,\alpha f_{q_1})-f^Q_2(T(P))-f^A_1(T(P))\\
	>{}&
	\beta_\alpha+\alpha\left(f^A_2(P_{q_1})+f^A_3(T(P_{q_1}))+f^Q_3(T(P_{q_1}))\right).
\end{split}
\end{equation}
This contradicts the inequality of Theorem \ref{mainthm}, hence there cannot be an ontic model with an $\alpha$ such that $\omega(f_P,\alpha f_{q_1})>E(\alpha)$.
	
The best way to look for a fragment of quantum mechanics that can be used to test Theorem \ref{mainthm}, is by restricting attention to fragments in which the predicted values for $f^A_2(P_{q_1})$, $f^A_3(T(P_{q_1}))$ and $f^Q_3(T(P_{q_1}))$ are zero, i.e., in which \eqref{unwaras} holds.
Experimentally, these values will of course always be greater than zero and the slope of $E_i$ in Figure \ref{ab-exp-label} is also positive.
But at least for the current choice of the fragment we have that the slope becomes closer to zero as the precision of the measurements increases, with the idealized case represented by the line $E_1$.

Since experimentally the slope of the line $E_i$ will never actually be zero, the best theoretical option we have for ruling out as much models as possible is by maximizing the distance between $f_1^Q(P)$ and $f_2^Q(T(P))+f_1^A(T(P))$.
A numerical analysis in $\mathbb{R}^3$ gives the value 0.236 for this maximal value.
A concrete fragment of quantum mechanics that comes close to this value is given by
\begin{equation} 							
\begin{IEEEeqnarraybox}[][c]{lll}
	\ket{q_1}=\cv{1}{0}{0},\quad&
	\ket{a_1}=\frac{1}{6}\sqrt{6}\cv{2}{1}{-1},\quad&	
	U=\begin{pmatrix}
		\tfrac{1}{2}\sqrt{2} & -\tfrac{1}{2}\sqrt{2} & 0\\
		\tfrac{1}{2}\sqrt{2} & \tfrac{1}{2}\sqrt{2} & 0\\
		0 & 0 & 1
	\end{pmatrix},\\
	\ket{q_2}=\cv{0}{1}{0},\quad&
	\ket{a_2}=\frac{1}{2}\sqrt{2}\cv{0}{1}{1},\quad&	
	\ket{\psi}=\frac{1}{4}\cv{1+\sqrt{3}}{1-\sqrt{3}}{2\sqrt{2}},\\
	\ket{q_3}=\cv{0}{0}{1},\quad&
	\ket{a_3}=\frac{1}{3}\sqrt{3}\cv{-1}{1}{-1},\quad&
	U\ket{\psi}=\frac{1}{4}\sqrt{2}\cv{\sqrt{3}}{1}{2}.		
\end{IEEEeqnarraybox}
\end{equation}
For this set we have
\begin{equation}
	\born{q_1}{\psi}-\born{q_2}{U\psi}-\born{a_1}{U\psi}=\tfrac{1}{48}(10\sqrt{3}-7)\approx 0.215.
\end{equation}
	
It is worthwhile to delve a bit into the required accuracy for an experiment in order for it to rule out at least some models.
Looking again at Figure \ref{ab-exp-label} we find that the eigenpreparation mixing models are ruled out whenever $E(f_1^Q(P))<f_1^Q(P)$, i.e., when
\begin{equation}\label{EMcriterion}
	f^Q_2(T(P))+f^A_1(T(P))	+f^Q_1(P)\left(f^A_2(P_{q_1})+f^A_3(T(P_{q_1}))+f^Q_3(T(P_{q_1}))\right)<f^Q_1(P).
\end{equation}

The measured relative frequencies will of course deviate from the predicted quantum mechanical probabilities.
To get a view of the amount of deviation that is permissible, we consider a worst case scenario where every relative frequency deviates an amount $\epsilon$ from the predicted probability in the direction that is most problematic for ruling out any models.
We then look at the maximal value for $\epsilon$ that would still allow us to conclude that eigenpreparation mixing models are ruled out.
Any lower value for $\epsilon$ would then of course allow us to rule out more models.
Making use of \eqref{EMcriterion}, we find that $\epsilon$ should satisfy
\begin{multline}
	\pee(q_2|Q,T,P)+\pee(a_1|A,T,P)+2\epsilon\\
	+\left(\pee(q_1|Q,P)-\epsilon\right)
	\left(\pee(a_2|A,P_{q_1})+\pee(a_3|A,T,P_{q_1})+\pee(q_3|Q,T,P_{q_1})+3\epsilon\right)\\
	<\pee(q_1|Q,P)-\epsilon.
\end{multline}
When using the numbers from the above example, we find that this constraint is satisfied whenever $\epsilon<5.06\%$.
This means that, roughly, eigenpreparation mixing models can be ruled out with an experimental setup with a fidelity of at least 95\%.
This seems already an experimentally feasible value \citep{Knee12}, while we have to bear in mind that we considered a very pessimistic scenario.
We may thus expect that experiments ruling out macroscopic realist models can be carried out in the near future, without having to rely on an additional assumption of non-invasive measurability.


\section{Discussion}\label{Discus}

Our analysis of macroscopic realism relied heavily on the framework of PTM models and ontic models.
We adopted this framework to have an analysis that is to a large extent theory independent, mimicking the derivation and predictions for the experimental violations of Bell type inequalities in this sense.
Obviously, though, there is a connection with quantum mechanics in the background.
We use quantum mechanics to make predictions about what kind of preparations, transformations and measurements are physically possible.
Theory independence is re-obtained by ultimately verifying experimentally that these operations are indeed physically realizable.
Consequently, we can only make use of small \emph{finite} fragments of quantum mechanics in our analysis.
This is reflected in Theorem \ref{mainthm}, which only makes use of two preparations, one transformation and two measurements.
However, implicit further assumptions lurk in the background, and we shall discuss and elucidate those here.

It seems innocent enough to assume that the PTM model under consideration contains many elements (preparations, transformations and measurements) apart from the ones that will show up in experiments.
Most of the time these elements are just coming along for a free ride.
This can again be compared to the case of Bell tests.
Of course not all systems display non-local behavior.
The point of the tests is to show that there are finite sets of preparations and measurements that experimentally violate Bell inequalities.
The conclusion is then that any ontic model for any PTM model that incorporates these preparations and measurements must be non-local \emph{irrespective} of what other elements may be incorporated in the PTM model.

It is then important that any additional elements that may be present in a PTM model, play no role in our analysis whatsoever.
The simplest way to ensure this is to not mention these elements.
We, however, have not adhered to this credo everywhere.
An explicit example is in our Definition \ref{lmrdef}.
To see if a particular ontic model is $(\alpha,\beta)$-supported, one has to quantify over \emph{all} preparations in the PTM model.
However, what exactly is the set of all preparations is of course a theory dependent question.
Consequently, the question if nature allows $(\alpha,\beta)$-support is also theory dependent.
The quantification is unproblematic because we are interested in \emph{ruling out} $(\alpha,\beta)$-support, instead of showing that it holds.
For this it suffices to look at a finite set of preparations, transformations and measurements.
This can again be compared to the case of locality.
To show that Bell inequalities are violated, a finite fragment of quantum mechanics suffices.
But the related question if nature satisfies the Tsirelson bound cannot be answered experimentally, for it would require to verify that the bound is satisfied for \emph{all} preparations and measurements.

A more troublesome quantification over the elements of the PTM model occurred in Section \ref{QMnotEM}, where we discussed the relation between the Leggett-Garg inequality and eigenpreparation mixing models.
We showed that, for eigenpreparation mixing models, to check that a transformation is non-invasive, one only has to verify that the transformation is non-invasive for eigenpreparations.
However, showing this itself is not unproblematic.
For example, to show that $T(P_q)$ and $P_q$ are operationally equivalent, one has to check that
\begin{equation}\label{opeqeq}
	\pee(m|M,T,P_q)=\pee(m|M,P_q)
\end{equation}
for \emph{every} measurement $M$ and \emph{every} outcome $m$.
In practice, this is not feasible.
To solve this problem one can in addition assume the existence of a finite tomographically complete set of measurements as one does, for example, in tests of non-contextuality (\citealp{Spekkens09}; \citealp{Mazurek15}; \citealp{KunjwalSpekkens15}) or, less explicitly, in the recent Leggett-Garg test in \citep{Knee16}.
Then, to verify \eqref{opeqeq}, one only has to check that it holds for the tomographically complete set.
However, such an assumption is ungrounded without assuming the (partial) validity of some theory \citep[\S 4.4]{Hermens11}.

Our analysis of macroscopic realism does not face this problem.
Nowhere do we assume that any two procedures are operationally equivalent and so there is no need to experimentally verify such an assumption.
This is another way in which we improve on the Leggett-Garg result.
However, there is a related issue, leading to a qualifying remark on the logical limitations for experimentally discriminating between the different types of macroscopic realism.

We argued in Section \ref{OEMsec} that macroscopic realism and non-invasive measurability implies the existence of an eigenpreparation mixing ontic model.
In short, the argument was that, with the use of non-invasive measurements of $Q$, any preparation $P$ can be turned into an eigenpreparation of $Q$ by measuring $Q$.
Because the measurement is non-invasive, the ontic state is unaltered and $P$ can therefore be written as a mixture of these eigenpreparations.

But, even if this possibility is ruled out for a given PTM model, there always remains the possibility that a new non-invasive $Q$ measurement procedure could be added to the model.  The post-measurement preparations $P_q$, produced by such a measurement would represent new eigenpreparations, not contained within our original PTM model. So eigenpreparation mixing can only be ruled out for a given PTM model, and there is always the logical possibility of extending the model to include more eigenpreparations, thereby restoring the possibility of an eigenpreparation mixing ontic model.

In this sense, our results are not completely theory independent.
However, this doesn't affect the main thrust of our result.  
Namely, given the eigenpreparations that we currently know, if we want to have an ontic model, other preparations (corresponding to superpositions in quantum mechanics) compel us to introduce novel ontic states.  If the model is to be augmented with new eigenpreparations and measurements to recover eigenpreparation mixing, then at some point these augmented models must deviate from the predictions of quantum theory.

The situation may be compared to the results on the $\psi$-ontic/$\psi$-epistemic divide.\footnote{(\citealp{PBR12}; \citealp{BCLM14}; \citealp{Leifer14}).}
These results show that an epistemic interpretation cannot fully explain the indistinguishability of non-orthognal quantum states.
But this indistinguishability itself is not a given theory-independent fact.
There is the logical possibility that by going beyond quantum mechanics there are measurements that \emph{can} distinguish quantum states with a single shot.
Similarly, our results show that macroscopic realism does nothing to explain the peculiar nature of superpositions.
But there is still the logical possibility that by going beyond quantum mechanics there are preparations in terms of which superpositions can be understood as mixtures after all.

\section{Conclusion}\label{conclusie}

The Leggett-Garg inequality is the best known constraint on macroscopic realism, but its significance has been diminished both by its reliance on the assumption of non-invasive measurability, and by the existence of known counterexamples, such as the de Broglie-Bohm and Kochen-Specker models, which are able to violate the Leggett-Garg inequality while being macro-realist about the relevant observables.

\citet{MaroneyTimpson16} clarified the different kinds of macroscopic realism possible, drawing a distinction between the counterexamples and the types of macroscopic realist models ruled out by Leggett-Garg inequality violations.  
However, despite the work of \citet{Allen16} extending the range of models which were in conflict with quantum theory, the distinctions introduced between eigenpreparation support and eigenpreparation undermining models had finite precision loopholes, and so these distinctions could not be empirically tested.

In this paper we have reanalyzed the difference between eigenpreparation supported and eigenpreparation undermining models, to look for their qualitative features which are robust against small variations in the model.  We defined a qualitatively eigenpreparation supported model as one in which there are no regions of the ontic state space in which all eigenpreparations have a small support, but at least one preparation has a large support.  We introduced the concept of $(\alpha,\beta)$-supported models to parameterize this feature: qualitatively eigenpreparation supported models are models for which there exist low values for $\alpha$ and $\beta$ such that the model is $(\alpha,\beta)$-supported.

We then showed that macroscopic realist models had $(\alpha,\beta)$-support curves which could be compared to empirical data, to rule out classes of macroscopic realist models for quantum theory.
We showed that eigenpreparation mixing models, the only kind that could also be ruled out by Leggett-Garg inequality violations, could be ruled out at relatively modest experimental errors.
However, we can also go beyond that, and rule out qualitatively eigenpreparation supported models.
As the precision of experimental tests of quantum theory increases, progressively less qualitatively eigenpreparation supported models are possible.
In the limit of noise free experimental data, all eigenpreparation supported models are ruled out, recovering the noise-free result.
As an additional feature, we note that many eigenpreparation undermining models may also be ruled out.
In contrast to the Leggett-Garg inequality violation, no troublesome assumption of non-invasive measurability is needed for any of these experimental tests.

\section*{Acknowledgments}
We would like to thank John-Mark Allen, Andrew Briggs, George Knee, Anna Pearson and Chris Timpson for useful discussions and comments on earlier drafts of this paper.
This project/publication was made possible through the support of a grant from Templeton Religion Trust.
The opinions expressed in this publication are those of the authors and do not necessarily reflect the views of Templeton Religion Trust.


\begin{appendix}

\section{Proof of Theorem \ref{mainthm}}

In this appendix we give a proof of Theorem \ref{mainthm}.
The proof makes use of the following lemma.
\begin{lemma}\label{mainlem}
Let $(\alpha,\beta)\in[0,\infty)\times[0,1]$ and consider an ontic model with a preparation $P$ that is $(\alpha,\beta)$-supported on $\Lambda_q$ for some value $q$ for the macro-observable $Q$.
Let $g:\Lambda\to[0,1]$ be a measurable function. 
Then, for every probability density $f_P$ that models the preparation $P$ 
\begin{equation}
	\int_{\Lambda_q}g(\lambda)f_P(\lambda)\dee\lambda\leq\sup_{f_q\in\Pi_q}\int_{\Lambda_q}g(\lambda)\alpha f_q(\lambda)\dee\lambda+\beta.
\end{equation} 
\end{lemma}
\begin{proof}
For the proof we make use of the following notation for the minimum:
\begin{equation}
	\left(f\wedge g\right)(\lambda):=\min\left(f(\lambda),g(\lambda)\right).
\end{equation}
We begin with a simple estimate making use of the definition of $(\alpha,\beta)$-support.
\begin{equation}
\begin{split}
	\int_{\Lambda_q}g(\lambda)f_P(\lambda)\dee\lambda
	&=
	\int_{\Lambda_q}f_P(\lambda)\dee\lambda
	-\int_{\Lambda_q}\left(1-g(\lambda)\right)f_P(\lambda)\dee\lambda\\
	&\leq
	\sup_{f_q\in\Pi_q}\int_{\Lambda_q}\left(f_P\wedge\alpha f_q\right)(\lambda)\dee\lambda+\beta
	-\int_{\Lambda_q}\left(1-g(\lambda)\right)f_P(\lambda)\dee\lambda\\
	&\leq
	\sup_{f_q\in\Pi_q}\int_{\Lambda_q}\left(f_P\wedge\alpha f_q\right)(\lambda)\dee\lambda+\beta
	-\int_{\Lambda_q}\left(1-g(\lambda)\right)(f_P\wedge f)(\lambda)\dee\lambda.
\end{split}
\end{equation}
The last estimate holds for any function $f$.
So in particular it holds for $f=\alpha \tilde{f}_q$ for every $\tilde{f}_q\in\Pi_q$, and we can take the supremum over all elements of $\Pi_q$:
\begin{equation}
	\int_{\Lambda_q}g(\lambda)f_P(\lambda)\dee\lambda
	\leq
	\sup_{f_q\in\Pi_q}\int_{\Lambda_q}\left(f_P\wedge\alpha f_q\right)(\lambda)\dee\lambda+\beta
	-\sup_{\tilde{f}_q\in\Pi_q}\int_{\Lambda_q}\left(1-g(\lambda)\right)(f_P\wedge\alpha\tilde{f}_q)(\lambda)\dee\lambda.
\end{equation}
Finally, we make use of the fact that $\sup_{x\in X}g(x)-\sup_{x'\in X}h(x')\leq\sup_{x\in X}(g(x)-h(x))$ to obtain
\begin{equation}
	\int_{\Lambda_q}g(\lambda)f_P(\lambda)\dee\lambda
	\leq
	\sup_{f_q\in\Pi_q}\left(\int_{\Lambda_q}\left(f_P\wedge\alpha f_q\right)(\lambda)\dee\lambda+\beta
	-\int_{\Lambda_q}\left(1-g(\lambda)\right)(f_P\wedge\alpha f_q)(\lambda)\dee\lambda\right),
\end{equation}
which is the desired result.
\end{proof}

\begin{proof}[Proof of Theorem \ref{mainthm}]
Let $P$ be any preparation that is $(\alpha,\beta)$-supported on $\Lambda_{q_1}$ and let $T$ be any transformation.
Throughout the proof we let $\gamma\in\Gamma_T$ be fixed.
The first estimate is based on the fact that the $Q=q_1$ states that transform to $Q=q_2$ states under $T$ form a subset of all the states that transform to $Q=q_2$ states under $T$. 
\begin{equation}
\begin{split}
	\pee(q_1|Q,P)-\pee(q_2|Q,T,P)
={}& 
    	\int_{\Lambda_{q_1}}f_P(\lambda)\dee\lambda
    	-\int_\Lambda\int_{\Lambda_{q_2}}\gamma(\dee\lambda'|\lambda)f_P(\lambda)\dee\lambda\\
={}&
    	\int_{\Lambda_{q_1}}\sum_{i=1}^3\int_{\Lambda_{q_i}}\gamma(\dee\lambda'|\lambda)f_P(\lambda)\dee\lambda
    	-\int_\Lambda\int_{\Lambda_{q_2}}\gamma(\dee\lambda'|\lambda)f_P(\lambda)\dee\lambda\\
\leq{}&
    	\int_{\Lambda_{q_1}}\int_{\Lambda_{q_1}}\gamma(\dee\lambda'|\lambda)f_P(\lambda)\dee\lambda
    	+\int_{\Lambda_{q_1}}\int_{\Lambda_{q_3}}\gamma(\dee\lambda'|\lambda)f_P(\lambda)\dee\lambda
\end{split}
\end{equation}

In the noise-free case the second term in this final expression would be zero because there we assumed that the probability for any $Q=q_1$ state to transform to a $Q=q_3$ state is zero.
In the noise-tolerant case we need the notion of $(\alpha,\beta)$-support to constraint the term, which will be done by invoking Lemma \ref{mainlem}.
However, a stronger estimate is obtained if we postpone this invocation until we have a better estimate for the first term.
For this term we have
\begin{align}
\begin{split}
	\MoveEqLeft[6] \int_{\Lambda_{q_1}}\int_{\Lambda_{q_1}}\gamma(\dee\lambda'|\lambda)f_P(\lambda)\dee\lambda-\pee(a_1|A,T,P) \\ 
={}&
    	\int_{\Lambda_{q_1}}\int_{\Lambda_{q_1}}\gamma(\dee\lambda'|\lambda)f_P(\lambda)\dee\lambda
    	-\int_\Lambda\int_\Lambda\xi_A(a_1|\lambda')\gamma(\dee\lambda'|\lambda)f_P(\lambda)\dee\lambda\\
={}&
    	\int_{\Lambda_{q_1}}\int_{\Lambda_{q_1}}\sum_{i=1}^3\xi_A(a_i|\lambda')\gamma(\dee\lambda'|\lambda)f_P(\lambda)\dee\lambda
    	-\int_\Lambda\int_\Lambda\xi_A(a_1|\lambda')\gamma(\dee\lambda'|\lambda)f_P(\lambda)\dee\lambda\\
\leq{}&
	\int_{\Lambda_{q_1}}\int_{\Lambda_{q_1}}\xi_A(a_2|\lambda')\gamma(\dee\lambda'|\lambda)f_P(\lambda)\dee\lambda
	+\int_{\Lambda_{q_1}}\int_{\Lambda_{q_1}}\xi_A(a_3|\lambda')\gamma(\dee\lambda'|\lambda)f_P(\lambda)\dee\lambda.
\end{split}
\end{align}

Combining these two estimates we find that
\begin{equation}\label{mainest}
	\pee(q_1|Q,P)-\pee(q_2|Q,T,P)-\pee(a_1|A,T,P)  
\leq
	\int_{\Lambda_{q_1}}\left(\int_{\Lambda_{q_3}}\gamma(\dee\lambda'|\lambda)+\int_{\Lambda_{q_1}}\xi_A(\{a_2,a_3\}|\lambda')\gamma(\dee\lambda'|\lambda)\right)f_P(\lambda)\dee\lambda.
\end{equation}	
Now note that
\begin{equation}
\begin{split}
	\int_{\Lambda_{q_3}}\gamma(\dee\lambda'|\lambda)+\int_{\Lambda_{q_1}}\xi_A(\{a_2,a_3\}|\lambda')\gamma(\dee\lambda'|\lambda)
\leq{}&
	\int_{\Lambda_{q_3}}\gamma(\dee\lambda'|\lambda)+\int_{\Lambda_{q_1}}\gamma(\dee\lambda'|\lambda)\\
\leq{}&
	\int_\Lambda\gamma(\dee\lambda'|\lambda)=1.
\end{split}
\end{equation}
This means that we can apply Lemma \ref{mainlem} to \eqref{mainest} to obtain the final estimate:
\begin{align}
\begin{split}
	\MoveEqLeft[6] \pee(q_1|Q,P)-\pee(q_2|Q,T,P)-\pee(a_1|A,T,P) \\ 
\leq{}&
	\sup_{f_{q_1}\in\Pi_{q_1}}\int_{\Lambda_{q_1}}\left(\int_{\Lambda_{q_3}}\gamma(\dee\lambda'|\lambda)+\int_{\Lambda_{q_1}}\xi_A(\{a_2,a_3\}|\lambda')\gamma(\dee\lambda'|\lambda)\right)\alpha f_{q_1}(\lambda)\dee\lambda+\beta\\
={}&
	\alpha\left(\pee(q_3|Q,P_{q_1})+\pee(\{a_2,a_3\}|A,T,P_{q_1})\right)+\beta.
\end{split}
\end{align}
\end{proof}

\end{appendix}

\begingroup
\sloppy
\printbibliography[heading=bibintoc]
\endgroup

\end{document}